\documentclass[journal,10pt]{IEEEtran}
\IEEEoverridecommandlockouts
\usepackage{cite}
\usepackage{indentfirst}
\usepackage{graphicx}
\usepackage{changepage}
\usepackage{stfloats}
\usepackage{amsfonts,amssymb}
\usepackage{bm}
\usepackage{algorithm}
\usepackage{algorithmic}
\usepackage{amsmath}
\usepackage{amsmath,amssymb,amsfonts}
\usepackage{times}
\usepackage{url}
\usepackage{cite}
\usepackage{textcomp}
\usepackage{xcolor}
\usepackage{color}
\usepackage{setspace}
\usepackage{enumitem}
\usepackage{float}
\usepackage{diagbox}
\usepackage[T1]{fontenc}
\usepackage[utf8]{inputenc}
\usepackage{threeparttable}
\usepackage{booktabs}
\usepackage{footnote}
\usepackage{graphicx}
\usepackage{pifont}
\usepackage{subfigure}
\usepackage{mathrsfs}
\newtheorem{pos}{Proposition}
\newtheorem{Corollary }{Corollary }
\newenvironment{proof}{{\indent \indent \it Proof:}}{\hfill $\blacksquare$\par}
\newtheorem{Lemma}{Lemma}
\newtheorem{Theorem}{Theorem}
\def\BibTeX{{\rm B\kern-.05em{\sc i\kern-.025em b}\kern-.08em
    T\kern-.1667em\lower.7ex\hbox{E}\kern-.125emX}}

\renewcommand{\theenumi}{\alph{enumi}}

\allowdisplaybreaks [4]
\begin{document}

\title{Engineering Favorable Propagation: Near-Field IRS Deployment for Spatial Multiplexing}

% \author{\IEEEauthorblockN{Yuxuan~Chen, Qingqing~Wu, \textit{Senior Member, IEEE}, Guangji~Chen,
% Qiaoyan~Peng},\and and Wen~Chen, \textit{Senior Member, IEEE}
\author{%
\IEEEauthorblockN{%
Qingqing~Wu, \textit{Senior Member, IEEE}, Yuxuan~Chen, Guangji~Chen,
Qiaoyan~Peng,}\\%
and Wen~Chen, \textit{Senior Member, IEEE}%
\vspace{-15pt}
\thanks{Q. Wu's work is supported by National Key R\&D Program of China (2023YFB2905000) and NSFC 62371289. Qingqing~Wu, Yuxuan~Chen, and Wen~Chen are with the Department of Electronic Engineering, Shanghai Jiao Tong University, Shanghai 200240, China (e-mail: {yuxuanchen@sjtu.edu.cn; qingqingwu@sjtu.edu.cn; wenchen@sjtu.edu.cn}). Guangji Chen is with the School of Electronic and Optical Engineering, Nanjing University of Science and Technology, Nanjing 210094, China  and
also with National Mobile Communications Research Laboratory, Southeast University (email:{guangjichen@njust.edu.cn}). Qiaoyan Peng is with the State Key Laboratory of Internet of Things for Smart City, University of Macau, Macao 999078, China, and also with the Department of Electronic Engineering, Shanghai Jiao Tong University, Shanghai 200240, China (e-mail: {qiaoyan.peng@connect.um.edu.mo}).}
}

\maketitle
\begin{abstract}
In intelligent reflecting surface (IRS)-assisted multiple-input multiple-output (MIMO) systems, a strong line-of-sight (LoS) link is required to compensate for the severe cascaded path loss. However, such a link renders the effective channel highly rank-deficient and fundamentally limits spatial multiplexing. To overcome this limitation, this paper leverages the large aperture of sparse arrays to harness near-field spherical wavefronts, and establishes a deterministic deployment criterion that strategically positions the IRS in the near-field of a base station (BS). This placement exploits the spherical wavefronts of the BS–IRS link to engineer decorrelated channels, thereby fundamentally overcoming the rank-deficiency issue in far-field cascaded channels. Based on a physical channel model for the sparse BS array and the IRS, we characterize the rank properties and inter-user correlation of the cascaded BS–IRS–user channel. We further derive a closed-form favorable propagation metric that reveals how the sparse array geometry and the IRS position can be tuned to reduce inter-user channel correlation. The resulting geometry-driven deployment rule provides a simple guideline for creating a favorable propagation environment with enhanced effective degrees of freedom. The favorable channel statistics induced by our deployment criterion enable a low-complexity maximum-ratio transmission (MRT) precoding scheme. This serves as the foundation for an efficient algorithm that jointly optimizes the IRS phase shifts and power allocation based solely on long-term statistical channel state information (CSI). Simulation results validate the effectiveness of our deployment criterion and demonstrate that our optimization framework achieves significant performance gains over benchmark schemes.
\end{abstract}% still approximation
\vspace{-3mm}
% }

\section{Introduction}
The pursuit of unprecedented spectral efficiency and ultra-high spatial resolution for sixth-generation (6G) systems has spurred the emergence of extremely large-scale multiple-input multiple-output (XL-MIMO) as a foundational technology \cite{r1,10379539,9903389,10068140,10545312}. By deploying antenna arrays with massive apertures, often operating at millimeter-wave or Terahertz frequencies, the communication environment inherently transitions from the conventional far-field to the more complex near-field region \cite{r21}. This near-field regime is characterized by a spherical electromagnetic wavefront, a fundamental shift that imbues the channel with a rich, distance-dependent spatial structure\cite{r8}. The resulting additional spatial dimension allows for user separation based on not only angle but also range, creating new degrees of freedom (DoF) for communication \cite{r6}. Harnessing these near-field effects offers a new paradigm for enhancing spatial multiplexing and diversity, establishing near-field XL-MIMO as a key enabler for the ambitious performance targets of 6G\cite{10944643}.

However, realizing these benefits in practice is challenged by hardware constraints, cost, and power limitations. To address these challenges, sparse array architectures, where antenna elements are spaced wider than the traditional half-wavelength distance, have been introduced as an efficient means to expand the effective aperture without proportionally increasing hardware complexity \cite{r8}. Sparse uniform arrays (SUAs) can significantly enlarge the near-field region and enhance spatial focusing capability, making them particularly suitable for near-field communications. Recent near-field sparse array studies\cite{r7,10681123,10930389,zhou2024sparse} further show that enlarging the effective array aperture via increased element spacing can extend the radiative near-field coverage and improve spatial resolution at essentially unchanged hardware cost. Nonetheless, the enlarged spacing substantially reshapes the array’s spatial response and maps users with distinct physical locations onto highly coupled effective channels, which leads to severe inter-user interference (IUI)\cite{10563980}. Consequently, interference suppression in sparse MIMO systems remains a pressing research challenge, as conventional linear precoding and far-field-based designs, which rely on well-separated angular signatures, fail to cope with the spatial aliasing and channel correlation intrinsic to sparse geometries\cite{10829761}.

Intelligent reflecting surfaces (IRSs) have emerged as a transformative technology for reengineering the wireless propagation environment \cite{r10,r24,wu2025intelligent,10382681,9722711,10475359}. By controlling the phase shifts of a large number of low-cost passive elements, the IRS can manipulate signal reflections to constructively enhance desired links and destructively suppress interference\cite{r22}. %%
A significant body of research has focused on exploiting this programmability for multiuser interference management in MIMO networks. To this end, the work \cite{r16} proposed a joint active and passive beamforming design to minimize interference leakage. Building upon these interference mitigation principles, the authors of \cite{r12} established a deterministic interference subspace alignment (ISA) framework which achieves the full multiplexing gain by completely eliminating interference. However, such methods often rely on instantaneous channel state information (CSI), which can impose prohibitive estimation overhead. The practical challenge of this overhead motivated subsequent research into designs based on slower-varying statistical CSI. In this domain, \cite{r20} demonstrated that such approaches can optimize long-term performance metrics, such as ergodic sum-rate, with significantly reduced feedback, and \cite{r23} further analyzed the associated power scaling laws and phase shift optimization. Beyond algorithmic design, the physical deployment and topology of IRSs have also been exploited as effective DoF for suppressing inter-stream and inter-user interference\cite{11007277}. Specifically, the authors in \cite{r14} proposed an orthogonal deployment of multiple IRSs, with spatially separated and oriented surfaces yielding orthogonal effective channels and mitigating mutual interference. Moreover, \cite{ShiJin-Channel} showed that jointly optimizing the IRS deployment and reflection design can customize the composite channel rank and singular-value profile.

Despite these advances, a fundamental gap remains in understanding and designing IRS assisted sparse MIMO systems. Most existing IRS works assume purely far-field propagation, in which the base station (BS)–IRS channel collapses to a rank one structure, leading to fully correlated user channels and severely limiting multi-user multiplexing \cite{r17}. Recent near-field studies \cite{r18,10098681} have started to model spherical wavefronts in XL-MIMO, revealing that wavefront curvature can significantly enhance DoF. Nevertheless, a critical open problem remains in translating this insight into practical design principles that can deterministically connect array sparsity and IRS placement with channel decorrelation. Furthermore, existing optimization frameworks typically depend on instantaneous CSI across the cascaded BS–IRS–user link, imposing prohibitive estimation overhead in large-scale deployments \cite{r20}.  Dynamic IRS adaptation as in \cite{partiIRS} may further improve performance, but at the cost of higher channel estimation and signaling overhead. Thus, there is an urgent need for geometry-driven, deterministic, and statistically robust frameworks capable of harnessing near-field effects, sparse apertures, and intelligent reflection control to achieve scalable interference suppression and performance maximization.

In this paper, we develop a unified theoretical and algorithmic framework for sparse IRS assisted multiuser MIMO systems operating in a near-field regime. Building on a rigorous physical channel model, we jointly address the deterministic channel correlation, IRS deployment geometry, and statistical beamforming optimization. The key contributions are summarized as follows:
\begin{itemize}
  \item First, we investigate the IRS-assisted multi-user communication system employing a hybrid sparse MIMO architecture, where a BS is equipped with a sparse uniform linear array (S-ULA) and the IRS is a sparse uniform planar array (S-UPA). To unveil the fundamental performance limits, we analyze the system from a spatial multiplexing perspective, utilizing the effective degree of freedom (EDoF)\cite{r29} as the evaluation metric. Specifically, we show that the enlarged apertures enabled by sparse arrays accentuate near-field spherical wavefronts on the BS–IRS link, thereby providing increased spatial resolution and improved user separability compared with the far-field regime. As a result, this sparse near-field configuration can achieve a substantially higher EDoF than conventional far-field IRS configuration in the LoS region of a BS.

  \item Next, we develop an analytical framework to characterize the rank properties of the cascaded BS–IRS–user channel. In particular, we demonstrate that while far-field planar waves lead to rank-deficient channels, strategic IRS placement in the near-field of a BS leverages spherical wavefronts to create statistically favorable propagation enhanced sub-channels, thereby overcoming the inherent rank-deficiency of far-field cascaded channels. Based on this analysis, we derive a deterministic deployment criterion that provides a fundamental physical design principle for engineering channel statistics and enabling spatial multiplexing gain.
  
  \item Finally, we propose an efficient algorithm to maximize the ergodic sum-rate by leveraging statistical CSI. The BS adopts low-complexity maximum-ratio transmission (MRT) precoding, while an alternating optimization procedure is developed to jointly design the IRS phase shifts and user power allocation based solely on long-term channel statistics. Notably, we derive closed-form solutions for the subproblems within the iterative procedure. Numerical results demonstrate that, under the proposed deployment and algorithm, inter-user interference is effectively suppressed and the achieved EDoF and ergodic sum-rate are substantially improved over conventional IRS configurations, approaching the interference-free benchmark.
\end{itemize}

The remainder of the paper is organized as follows. Section II presents the system and channel models. Section III analyzes far- and near-field cascaded channels and derives the deployment criterion. Section IV proposes the statistical-CSI algorithm for ergodic sum-rate maximization. Section V provides numerical validations and performance evaluations. Section VI concludes the paper. The main system symbols used throughout the paper are summarized in Table ~\ref{tab:notations}. 

\textit{Notations}: We use italic letters for scalars, boldface lower-case letters for vectors, and boldface upper-case letters for matrices. The operators $\mathrm{tr}(\cdot)$ and $|\cdot|_F$ stand for the trace and Frobenius norm, respectively. The operator $\mathrm{diag}(\mathbf{a})$ returns a diagonal matrix with the elements of vector $\mathbf{a}$ on its main diagonal. The statistical expectation and variance are denoted by $\mathbb{E}[\cdot]$ and $\mathrm{var}\{\cdot\}$. The function $\mathrm{gcd}(a,b)$ denotes the greatest common divisor of integers $a$ and $b$. Finally, $\mathcal{CN}(\boldsymbol{\mu}, \mathbf{\Sigma})$ denotes the circularly-symmetric complex Gaussian distribution with mean $\boldsymbol{\mu}$ and covariance matrix $\mathbf{\Sigma}$.
%Section II
\begin{table*}[!t]
\caption{ Table of Notations}
\label{tab:notations}
\centering

\scriptsize
\renewcommand{\arraystretch}{1.08}
\setlength{\tabcolsep}{4pt}
\begin{tabular}{p{0.17\textwidth}|p{0.77\textwidth}}
\hline
$M$ & Number of BS antennas in the sparse ULA. \\
\hline
$N=N_uN_v$ & Total number of IRS reflecting elements, where $N_u$ and $N_v$ are the numbers of elements along the $y$- and $z$-axes, respectively. \\
\hline
$K$ & Number of users. \\
\hline
$d_{\mathrm{BS}},\,d_{\mathrm{IRS}}$ & BS and IRS element spacings, respectively. \\
\hline
$\zeta_{\mathrm{BS}},\,\zeta_{\mathrm{IRS}}$ & Sparsity factors of the BS array and IRS, respectively. \\
\hline
$\mathbf{r}_{\mathrm{IRS}}=[l_x,l_y,l_z]^T,\; l=\|\mathbf{r}_{\mathrm{IRS}}\|$ & IRS center location and the BS-IRS distance. \\
\hline
$\mathbf{F}$ & Near-field BS-IRS channel matrix. \\
\hline
$\mathbf{g}_k$ & IRS-user channel vector for user $k$. \\
\hline
$\mathbf{h}_k$ & Cascaded BS-IRS-user channel vector for user $k$. \\
\hline
$\boldsymbol{\theta},\mathbf{\Theta}$ & IRS phase-shift vector and the corresponding diagonal reflection matrix. \\
\hline
$\mathbf{w}_k,\; p_k$ & MRT precoding vector and power-allocation factor for user $k$, respectively. \\
\hline
$\gamma_k,\; R_{\mathrm{erg}}$ & SINR of user $k$ and ergodic sum-rate, respectively. \\
\hline
$\varepsilon$ & Effective degree of freedom (EDoF) metric of the aggregate channel. \\
\hline
$\rho_{jk}$ & Spatial correlation coefficient between the cascaded channels of users $j$ and $k$. \\
\hline
$g_{jk}$ & Favorable-propagation metric that measures the normalized variance of $\mathbf{h}_j^H\mathbf{h}_k$. \\
\hline
$\Delta_{s,t}$ & Phase increment associated with the relative IRS-element lag pair $(s,t)$. \\
\hline
$\Delta$ & Fundamental phase increment under the symmetric deployment condition $\mu_y=\mu_z$. \\
\hline
$\alpha_k,\;\kappa_k,\;\beta_k$ & Large-scale path gain, Rician factor, and LoS power ratio of user $k$, respectively. \\
\hline
$C_k$ & NLoS contribution term in $\mathbb{E}[\mathbf{h}_k^H\mathbf{h}_k]$. \\
\hline
$C_{k,j}$ & NLoS-NLoS contribution term in $\mathbb{E}[|\mathbf{h}_k^H\mathbf{h}_j|^2]$. \\
\hline
$\mathbf{A}_{k,j}$ & Deterministic LoS coupling matrix appearing in the statistical channel moments. \\
\hline
$\mathbf{D}_k$ & Fourth-order statistical matrix appearing in the interference term. \\
\hline
\end{tabular}
\vspace{-10pt}
\end{table*}
\section{System Model}

As shown in Fig. \ref{fig1}, we consider a MIMO communication system, where a BS serves $K$ users with the assistance of an IRS. The system is modeled in a three-dimensional Cartesian coordinate system, with the BS–IRS link characterized by a near-field channel model and the IRS–user links characterized by a far-field channel model. Furthermore, we assume that the direct BS-user links are obstructed due to practical blockage. The BS is equipped with an S-ULA comprising $M$ antenna elements, where $M$ is assumed to be an odd number for notational convenience. The array is aligned along the $x$-axis and centered at the origin. The position vector of the $i$-th antenna is given by $\mathbf{p}_{\mathrm{BS},i}=[id_{\mathrm{BS}},0,0]^T$, where $i\in \{0,\pm 1,...,\pm (M-1)/2\}$ and the inter-antenna spacing is defined as $d_{\mathrm{BS}}=\zeta _{\mathrm{BS}}\frac{\lambda}{2}
$ with $\lambda$ denoting the carrier wavelength and 
$\zeta _{\mathrm{BS}}$ denoting the array sparsity factor at the BS. The IRS is an S-UPA comprising $N=N_u \times N_v$ reflecting elements, parallel to the $y-z$ plane. The numbers of elements along the $y$- and $z$-axes, $N_u$ and $N_v$, are both assumed to be odd. The geometric center of the IRS is located at $\mathbf{r}_{\mathrm{IRS}}=[l_x,l_y,l_z]^T$ with its distance from the origin being $l=\|\mathbf{r}_{\mathrm{IRS}}\|$. The position vector of the $(u_n, v_n)$-th IRS element is $\mathbf{p}_{\mathrm{IRS},n}=\mathbf{r}_{\mathrm{IRS}}+\boldsymbol{\delta}_{n}$, where $n\in \{1,2,...,N\}
$ and $\boldsymbol{\delta}_{n}=[0,u_n d_{\mathrm{IRS}},v_n d_{\mathrm{IRS}}]^T$ is the local displacement vector. Here, the pair $(u_n, v_n)$ denotes the two-dimensional coordinates corresponding to the $n$-th element, with $\left. u_n=\left. \lfloor \frac{n-1}{N_v} \right. \right. \rfloor -\frac{N_u-1}{2}\in \{0,\pm 1,\dots ,\pm \frac{N_u-1}{2}\}$ and $v_n=((n-1)\mathrm{mod}N_v)-\frac{N_v-1}{2}\in \{0,\pm 1,\dots ,\pm \frac{N_v-1}{2}\}$. The element spacing is defined as $d_{\mathrm{IRS}}=\zeta_{\mathrm{IRS}}\lambda/2$ with $\zeta_{\mathrm{IRS}}$ being the IRS sparsity factor. For analytical tractability, we consider a calibrated BS–IRS system with common timing and frequency references and known deployment geometry. Practical synchronization and geometry mismatches, as well as calibration errors, mutual coupling, and IRS hardware impairments \cite{zz1,zz2,zz3}, mainly introduce phase perturbations and do not fundamentally alter the near-field rank-enhancement mechanism. The channel between BS and IRS array is represented by $\mathbf{F}\in \mathbb{C} ^{M\times N}$ and the channel entry between the $i$-th BS antenna and the $n$-th IRS element is given by
\begin{align}
[\mathbf{F}]_{i,n}=\rho _{\mathrm{T}}e^{-j\frac{2\pi}{\lambda}d_{i,n}},
\end{align}
where $\rho _{\mathrm{T}}$ is the complex-valued reference channel gain and $d_{i,n}=\parallel \mathbf{p}_{\mathrm{IRS},n}-\mathbf{p}_{\mathrm{BS},i}\parallel$ denotes the distance between the $i$-th BS antenna to the $n$-th IRS element that can be further expressed as
\begin{align}
\!d_{i,n}\!=\!\sqrt{(l_x-id_{\mathrm{BS}})^2\!+\!(l_y+u_nd_{\mathrm{IRS}})^2\!+\!(l_z+v_nd_{\mathrm{IRS}})^2}.
\end{align}
\begin{figure}[!t]
    \centering
    \includegraphics[width=0.5\textwidth]{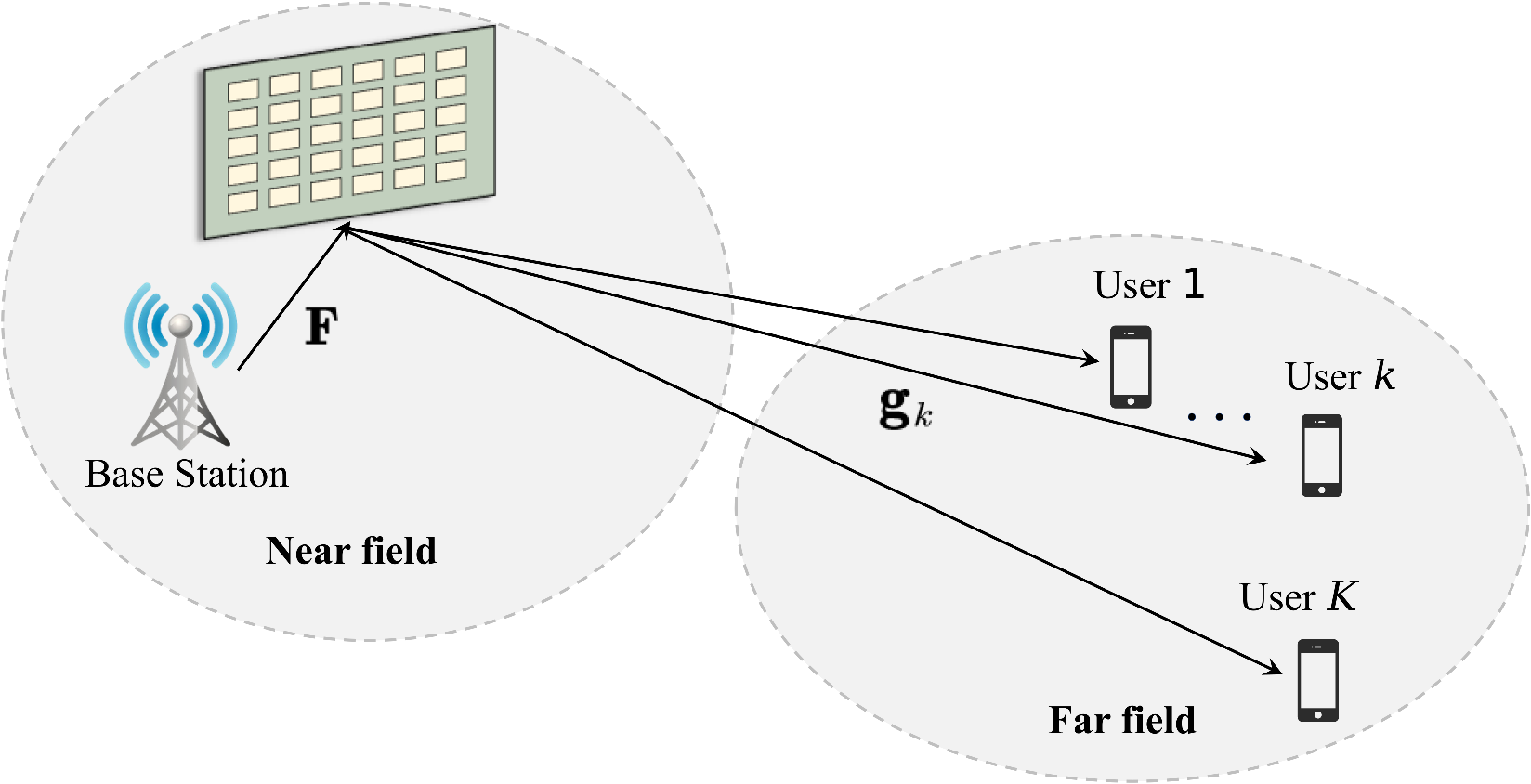}
    \caption{IRS assisted multi-user sparse MIMO communication.}
    \label{fig1}
\end{figure}
Then, the far-field channel from the IRS to the $k$-th user $\mathbf{g}_k\in \mathbb{C} ^{N\times 1}$ can be expressed as
\begin{align}\label{channelIRS2User}
\mathbf{g}_k=\rho _{\mathrm{R},k}\mathbf{a}_{\mathrm{S},k}\left( \Phi _{\mathrm{R},\mathrm{k}}^{\mathrm{AOD}},\Theta _{\mathrm{R},k}^{\mathrm{AOD}} \right),
\end{align}

\begin{figure*}[hb] 	
	\hrulefill
	\centering
	\begin{equation}
		\setlength\abovedisplayskip{2pt}
		\setlength\belowdisplayskip{2pt}
		\small
		\begin{aligned}
	d_{i,n} &= \sqrt{(l_x-id_{\mathrm{BS}})^2+(l_y+u_n d_{\mathrm{IRS}})^2+(l_z+v_n d_{\mathrm{IRS}})^2}  \\
    &= l\sqrt{1+\frac{\lambda}{l^2}(-il_x\zeta_{\mathrm{BS}}+u_n l_y\zeta_{\mathrm{IRS}}+v_n l_z\zeta_{\mathrm{IRS}})+\frac{\lambda^2}{4l^2}((i\zeta_{\mathrm{BS}})^2+(u_n\zeta_{\mathrm{IRS}})^2+(v_n\zeta_{\mathrm{IRS}})^2)}
		\end{aligned}
		\label{distance}
	\end{equation}	
\end{figure*}
where $\rho _{\mathrm{R},k}$ is the complex channel gain, $\Phi _{\mathrm{R},\mathrm{k}}^{\mathrm{AOD}}=\pi \zeta _{\mathrm{IRS}}\cos \theta _{\mathrm{R},k}^{\mathrm{AOD}}
$ and $\Theta _{\mathrm{R},k}^{\mathrm{AOD}}=\pi \zeta _{\mathrm{IRS}}\sin \phi _{\mathrm{R},k}^{\mathrm{AOD}}\sin \theta _{\mathrm{R},k}^{\mathrm{AOD}}$ with $\theta _{\mathrm{R},k}^{\mathrm{AOD}}$ and $\phi _{\mathrm{R},k}^{\mathrm{AOD}}$ denoting the horizontal angle of departure (AoD) and the vertical AoD at the IRS, respectively. Furthermore, the array response vector $\mathbf{a}_L$ of ULA can be unified by 
\begin{align}\label{ULA}
\mathbf{a}_L\left( X \right) =\left[ 1,e^{jX},...,e^{jX\left( L-1 \right)} \right] ^T.
\end{align}
It is worth noting that the array response  $\mathbf{a}_{\mathrm{S},k}$ for UPA can be decomposed into that of ULA as $\mathbf{a}_{\mathrm{S},k}\left( X,Y \right) =\mathbf{a}_{N_u}\left( X \right) \otimes \mathbf{a}_{N_v}\left( Y \right) $, where $\otimes$ is the Kronecker product. 
Moreover, the passive beamforming matrix of the IRS is denoted by $\mathbf{\Theta }=\mathrm{diag}\mathrm{(}e^{j\phi _1},e^{j\phi _2},\cdots ,e^{j\phi _N})$. Thus, the cascaded channel from the BS to the $k$-th user via the IRS is expressed as
\begin{align}
\mathbf{h}_k=\mathbf{F}\mathbf{\Theta g}_k.
\end{align}
Let $\boldsymbol{W}=\left[ \boldsymbol{w}_1,\boldsymbol{w}_2,\cdots,\boldsymbol{w}_K \right] \in \mathbb{C} ^{M\times K}$ denote the linear precoding matrix at the BS, where $\mathbf{w}_k \in \mathbb{C}^{M\times 1}$ is the precoding vector for user $k$. The baseband transmitted signal is given by $\boldsymbol{Ws}$, where $\boldsymbol{s}=\left[ s_1,\cdots,s_K \right] ^T$ is the data vector with each element $s_k$ being an independent variable with zero mean and normalized power. The received signal at user $k$ can be expressed as
\begin{align}
y_k=\boldsymbol{h}_{k}^{H}\boldsymbol{w}_ks_k+\sum_{j=1,j\ne k}^K{\boldsymbol{h}_{k}^{H}\boldsymbol{w}_j}s_j+\mathbf{z}_k,
\end{align}
where $\mathbf {z}_k\sim \mathcal {CN}(0,\sigma_k^2)$ is the additive white Gaussian noise at user $k$ with power $\sigma_k ^{2}$. Accordingly, the signal-to-interference-plus-noise ratio (SINR) at the user $k$ is given by
\begin{align}
\gamma _k=\frac{\left| \boldsymbol{h}_{k}^{H}\boldsymbol{w}_k \right|^2}{\sum_{j=1,j\ne k}^K{\left| \boldsymbol{h}_{k}^{H}\boldsymbol{w}_j \right|^2}+\sigma _{k}^{2}}.
\end{align}
Thus, the achievable sum rate of the system is given by
\begin{align}
R=\sum_{k=1}^K{R_k}\triangleq \sum_{k=1}^K{\log _2\left( 1+\gamma _k \right)}.
\end{align}

\section{IRS Deployment Strategy in Sparse MIMO}
\subsection{Channel Correlation Analysis in far-field}
Consider the scenario where the IRS is located in the far-field of the BS, which requires the distance $l$ to satisfy $l > \frac{2\left( D_{\mathrm{BS}}+D_{\mathrm{IRS}} \right) ^2}{\lambda}$, with $D_{\mathrm{BS}}$ and $D_{\mathrm{IRS}}$ denoting the apertures of the BS and IRS arrays, respectively \cite{r30}. Then, the spherical wavefronts emanating from the transmit antennas can be accurately approximated as planar waves at the receiver. The planar wave assumption is mathematically captured by a first-order Taylor series expansion of the propagation distance:
\begin{align}
   d_{i,n}\approx l-id_{\mathrm{BS}}\mu _x+u_nd_{\mathrm{IRS}}\mu _y+v_nd_{\mathrm{IRS}}\mu _z,
\end{align}
where $\mu_x=l_x/l$, $\mu_y=l_y/l$, and $\mu_z=l_z/l$ are the direction cosines of the IRS's center relative to the BS.

Under this approximation, the channel matrix $\mathbf{F}$ becomes rank-deficient, specifically rank-one. It can be expressed as
\begin{align}
\mathbf{F}\approx \rho _{\mathrm{T}}\mathbf{a}_M\left( \Theta _{\mathrm{T}}^{\mathrm{AOD}} \right) \mathbf{a}_{\mathrm{S}}^{H}\left( \Phi _{\mathrm{T}}^{\mathrm{AOA}},\Theta _{\mathrm{T}}^{\mathrm{AOA}} \right),
\end{align}
where $\Theta_{\mathrm{T}}^{\mathrm{AOD}}$ and $(\Phi_{\mathrm{T}}^{\mathrm{AOA}}, \Theta_{\mathrm{T}}^{\mathrm{AOA}})$ are determined by the AoD at the BS and the angles of arrival (AoA) at the IRS similarly in \eqref{channelIRS2User}, respectively.

To quantify the impact of this rank-one channel $\mathbf{F}$ on multi-user performance, we examine the spatial correlation between the cascaded channels of two arbitrary users, $j$ and $k$, defined by the correlation coefficient:
\begin{align}
\rho _{jk}=\frac{|\mathbf{h}_{j}^{H}\mathbf{h}_k|}{\parallel \mathbf{h}_j\parallel \parallel \mathbf{h}_k\parallel}.
\end{align}
With the IRS located in the far field of the BS, the cascaded channel for user $k$ simplifies to $\mathbf{h}_k\approx \chi _k\mathbf{a}_M(\Theta _{\mathrm{T}}^{\mathrm{AOD}})
$, where $\chi _k=\rho _{\mathrm{T}}\rho _{\mathrm{R},k}\mathbf{a}_{\mathrm{S}}^{H}\left( \Phi _{\mathrm{T}}^{\mathrm{AOA}},\Theta _{\mathrm{T}}^{\mathrm{AOA}} \right) \mathbf{\Theta a}_{\mathrm{S},k}\left( \Phi _{\mathrm{R},\mathrm{k}}^{\mathrm{AOD}},\Theta _{\mathrm{R},k}^{\mathrm{AOD}} \right)$. The correlation coefficient thus becomes: 
\begin{align}
\rho _{jk}\approx \frac{|\chi _{j}^{*}\chi _k\mathbf{a}_{M}^{H}\left( \Theta _{\mathrm{T}}^{\mathrm{AOD}} \right) \mathbf{a}_M\left( \Theta _{\mathrm{T}}^{\mathrm{AOD}} \right) |}{\left| \chi _j\chi _k \right|M}=1.
\end{align}

The result $\rho_{jk} \approx 1$ indicates that the users' channels are perfectly correlated in the far-field LoS scenario. It implies that the users are spatially indistinguishable from the BS's perspective, rendering linear precoding ineffective for multi-user spatial multiplexing. This fundamental limitation severely degrades system performance.

\subsection{Near-Field Channel Modeling and Analysis}
\begin{figure}[t]
\centering
\includegraphics[width=0.5\textwidth]{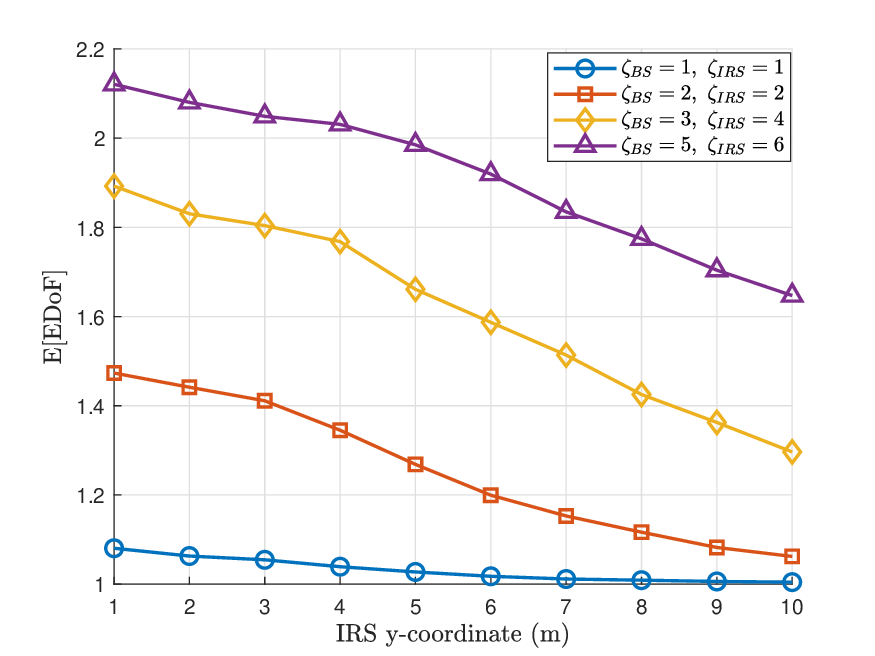}
\caption{$\mathbb{E} \left[ \mathrm{EDoF} \right]$ versus sparsity and IRS position when $N$=1600.}
\label{EDoFv1}
\vspace{-3pt}
\end{figure}
The fundamental limitation of the far-field model is that users’ spatial signatures become perfectly correlated, which renders multi-user multiplexing ineffective. This motivates a transition to near-field analysis, where the spherical nature of the wavefront is explicitly taken into account. In this regime, which is characterized by the condition $l<\frac{2\left( D_{\mathrm{BS}}+D_{\mathrm{IRS}} \right) ^2}{\lambda}$, the path-length variations across the array elements are no longer linear. 

 To accurately capture this wavefront curvature, we employ a Fresnel-type second-order Taylor expansion of the exact distance $d_{i,n}$ in \eqref{distance}, which is widely adopted in near-field spherical-wave modeling for extremely large-scale arrays \cite{lu1,r7}. Thus, $d_{i,n}$ can be approximated as 
\begin{align}\label{dbest}
&d_{i,n} \approx l \!+\! \frac{\lambda}{2}\left( -i\zeta_{\mathrm{BS}}\mu_x \!+\! u_n\zeta_{\mathrm{IRS}}\mu_y \!+\! v_n\zeta_{\mathrm{IRS}}\mu_z \right)\nonumber\\
&+ \!\!\frac{\lambda^2}{8l}\left[ (i\zeta_{\mathrm{BS}})^2(1-\mu_x^2) \!\!+\!\! (u_n\zeta_{\mathrm{IRS}})^2(1-\mu_y^2) \!\!+\!\! (v_n\zeta_{\mathrm{IRS}})^2(1-\mu_z^2) \right]\nonumber\\
&+\!\!\! \frac{\lambda^2}{4l}\left( iu_n\zeta_{\mathrm{BS}}\zeta_{\mathrm{IRS}}\mu_x\mu_y\!\! + \!\!iv_n\zeta_{\mathrm{BS}}\zeta_{\mathrm{IRS}}\mu_x\mu_z \!\!- \!\!u_n v_n\zeta_{\mathrm{IRS}}^2\mu_y\mu_z \right).
\end{align}

To assess the attainable spatial multiplexing gain, we adopt the EDoF as a performance metric. For an aggregate multi-user channel matrix $\mathbf{H}=[\mathbf{h}_1, \dots, \mathbf{h}_K]$, the EDoF is defined as \cite{4418491}
\begin{align}\label{EDoFDefine}
\varepsilon = \left( \frac{\mathrm{tr}(\mathbf{HH}^H)}{\|\mathbf{HH}^H\|_F} \right)^2,
\end{align}
where an EDoF value approaching the number of users $K$ signifies spatially orthogonal channels suitable for multiplexing. To assess the fundamental capabilities of the propagation environment, we evaluate the expected EDoF, $\mathrm{E}[\varepsilon]$, by averaging over numerous realizations of i.i.d. random IRS phase shifts.  A larger $\varepsilon$ indicates a more even eigenvalue distribution and hence a larger effective rank, implying better spatial separability among users. For a given aggregate channel power, it also corresponds to weaker inter-user correlation in the Gram matrix, meaning that the user channels are closer to being orthogonal. Therefore, $\varepsilon$ is a physically meaningful indicator of the spatial multiplexing capability of the cascaded channel. Under MRT, such channel decorrelation tends to reduce inter-user interference and improve the achievable sum-rate \cite{EDOF11}.

The impact of leveraging near-field effects through sparse arrays is illustrated in Fig. \ref{EDoFv1}. The results demonstrate that array sparsity acts as a critical enabler for spatial multiplexing. Specifically, the conventional dense array configuration ($\zeta_{\text{BS}}=1, \zeta_{\text{IRS}}=1$) yields an expected EDoF that remains close to 1, which reflects the rank-deficient nature of the channel in this regime. In contrast, increasing the sparsity factors substantially enhances the EDoF. This improvement stems from the expanded effective apertures, which accentuate the wavefront curvature and generate distinct spatial signatures for users. Furthermore, this gain diminishes as the IRS-BS distance increases, confirming that the high degrees of freedom are intrinsically tied to the near-field propagation conditions.

To rigorously characterize this inherent capability, we analyze the statistical properties of the cascaded channel under the assumption that the IRS phase shifts ${\phi_n}$ are independent and identically distributed (i.i.d.) random variables drawn from the uniform distribution $\mathcal{U}[0,2\pi)$.  This assumption is adopted as an analytical tool to decouple the geometry-induced propagation behavior from the specific IRS beamforming objective. In this way, the resulting analysis captures the intrinsic channel decorrelation and user-separation capability enabled by the proposed near-field sparse deployment, rather than the performance of any particular optimized phase configuration. This favorable channel structure can then be exploited by deterministic IRS phase designs.
\subsubsection{Expected Inter-User Correlation} 

To begin our analysis of the system's multi-user capabilities, we first evaluate the expected correlation between the channels of two distinct users, $j$ and $k$ as follows:
\begin{align}\label{ExpectUser}
\mathbb{E}[\mathbf{h}_{j}^{H}\mathbf{h}_k] = \mathbf{g}_{j}^{H} \mathbb{E}[\mathbf{\Theta}^H\mathbf{F}^H\mathbf{F}\mathbf{\Theta}] \mathbf{g}_k.    
\end{align}
Due to the statistical independence of the phase shifts, the expectation of the matrix product simplifies significantly. Specifically, for $n \neq n'$, the off-diagonal elements are zero since $\mathbb{E}[e^{-j\phi_n}e^{j\phi_{n'}}] = \mathbb{E}[e^{-j\phi_n}]\mathbb{E}[e^{j\phi_{n'}}] = 0$. The diagonal elements ($n=n'$) evaluate to one, i.e., $\mathbb{E}[e^{-j\phi_n}e^{j\phi_n}] = 1$. Consequently, the expected matrix becomes
\begin{align}
\mathbb{E}[\mathbf{\Theta}^H\mathbf{F}^H\mathbf{F}\mathbf{\Theta}] = \mathrm{diag}\left([\mathbf{F}^H\mathbf{F}]_{1,1}, \dots, [\mathbf{F}^H\mathbf{F}]_{N,N}\right).
\end{align}
Based on this, the expected inner product in \eqref{ExpectUser} thus reduces to
\begin{align}
\!\!\mathbb{E}[\mathbf{h}_{j}^{H}\mathbf{h}_k] \!\!=\!\!\sum_{n=1}^{N} [\mathbf{g}_j]_n^* [\mathbf{g}_k]_n \sum_{i=1}^{M} |[\mathbf{F}]_{i,n}|^2 \!\!=\! \!M|\rho_T|^2 (\mathbf{g}_j^H \mathbf{g}_k).
\end{align}
As the users are located in the far-field of the IRS, $\mathbf{g}_j$ and $\mathbf{g}_k$ are array response vectors corresponding to distinct spatial angles. For an IRS with a large aperture and sufficient angular separation between users, these vectors are nearly orthogonal. Consequently, the magnitude of their inner product $|\mathbf{g}_j^H \mathbf{g}_k|$ is significantly suppressed, making the expected correlation $\mathbb{E}[\mathbf{h}_{j}^{H}\mathbf{h}_k]$ negligible for $j \neq k$. This demonstrates that the user channels are, on average, weakly correlated, a foundational property for enabling multi-user communication.

\subsubsection{Favorable Propagation Analysis}

For robust multi-user communication, a powerful property is favorable propagation, which ensures that as the number of antennas grows the channel vectors of different users become deterministically orthogonal \cite{10769492}. Favorable propagation is formally achieved if the normalized variance of the channel inner product vanishes:
\begin{align}
g_{jk} \triangleq \frac{\mathrm{var}\left\{ \mathbf{h}_{j}^{H}\mathbf{h}_k \right\}}{\mathbb{E} \left[ \|\mathbf{h}_j\|^2 \right] \mathbb{E} \left[ \|\mathbf{h}_k\|^2 \right]} \to 0, \quad \text{for } j\neq k.    
\end{align}

The magnitude of $g_{jk}$ quantifies the severity of inter-user interference (IUI), with a vanishing $g_{jk}$ indicating asymptotic orthogonality. Since the behavior of $g_{jk}$ is governed by the variance term $\text{var}\{\mathbf{h}_{j}^{H}\mathbf{h}_{k}\}$, we derive its closed form in the following lemma by leveraging the random IRS phase shifts.

\begin{Lemma}
Assuming independent and uniformly distributed random IRS phase shifts, the variance of the inner product between the cascaded channels of distinct users $j$ and $k$ is given by
\begin{align}\label{var}
\!\!\!\!\!\!\mathrm{var}\{\mathbf{h}_{j}^{H}\mathbf{h}_k\}
\!\!= \!\!\sum_{n_1=1}^N \!\sum_{\substack{n_2=1,\\ n_2\ne n_1}}^N\!|[\mathbf{g}_j]_{n_1}|^2 |[\mathbf{g}_k]_{n_2}|^2
|[\mathbf{F}^H\mathbf{F}]_{n_1,n_2}|^2.
\end{align}

\begin{proof}
Please refer to Appendix A.
\end{proof}
\end{Lemma}

\begin{Theorem}
Under the near-field spherical-wavefront model, the favorable propagation metric $g_{jk}$ of the sparse MIMO system can be approximated as
\begin{align}\label{gjk}
g_{jk}\approx\frac{1}{M^{2}N^{2}}
\sum_{\substack{s=-(N_{u}-1)\\}}^{N_{u}-1}
\sum_{\substack{t=-(N_{v}-1)\\(s,t)\neq(0,0)}}^{N_{v}-1}
w(s,t)f_{M}(\Delta_{s,t}),
\end{align}
where $w(s,t)=(N_{u}-|s|)(N_{v}-|t|)$ and $f_{M}(x)=(\frac{\sin(Mx/2)}{\sin(x/2)})^{2}$ denotes the squared Dirichlet kernel of order $M$. The corresponding phase increment $\Delta_{s,t}$ is given by
\begin{align}\label{deltaNew}
\Delta_{s,t}\triangleq s\left(\frac{\pi\lambda}{2l}\zeta_{\text{BS}}\zeta_{\text{IRS}}\mu_{x}\mu_{y}\right)+t\left(\frac{\pi\lambda}{2l}\zeta_{\text{BS}}\zeta_{\text{IRS}}\mu_{x}\mu_{z}\right). 
\end{align}
\begin{proof}
From Lemma 1, $g_{jk}$ is determined by the squared modulus of the Gram-matrix element:
\begin{align}
|[\mathbf{F}^H\mathbf{F}]_{n_1,n_2}|^2=|\rho _T|^4\left| \sum_{i=-\frac{M-1}{2}}^{\frac{M-1}{2}}{e^{-j\frac{2\pi}{\lambda}(d_{i,n_1}-d_{i,n_2})}} \right|^2.
\end{align}
Substituting the second-order Taylor expansion of $d_{i,n}$ in \eqref{dbest} into the phase and taking the difference $d_{i,n_{1}}-d_{i,n_{2}}$, the terms depending only on $i$ cancel out, and the phase can be written as
\begin{align}
\frac{2\pi}{\lambda}(d_{i,n_{1}}-d_{i,n_{2}})\approx i\Delta_{n_{1},n_{2}}+\Phi_{n_{1},n_{2}}.    
\end{align}
The coefficient $\Delta_{n_{1},n_{2}}$ follows from the cross-term in \eqref{dbest}:
\begin{align}
\!\!\!\!\!\!\Delta _{n_1,n_2}\!=\!\frac{\pi \lambda}{2l}\zeta _{\mathrm{BS}}\zeta _{\mathrm{IRS}}\mu _x[\mu _y(u_{n_1}-u_{n_2})\!+\!\mu _z(v_{n_1}-v_{n_2})],   
\end{align}
which reduces to $\Delta_{s,t}$ as defined in the theorem by letting $s=u_{n_{1}}-u_{n_{2}}$ and $t=v_{n_{1}}-v_{n_{2}}$. The term $\Phi_{n_{1},n_{2}}$ collects all terms independent of $i$ from \eqref{dbest}:
\begin{align}
&\Phi_{n_{1},n_{2}}=\pi\zeta_{\text{IRS}}(\mu_{y}s+\mu_{z}t)+\frac{\pi\lambda}{4l}\zeta_{\text{IRS}}^{2}[(u_{n_{1}}^{2}-u_{n_{2}}^{2})(1-\mu_{y}^{2})\nonumber\\
&+(v_{n_{1}}^{2}-v_{n_{2}}^{2})(1-\mu_{z}^{2})-2\mu_{y}\mu_{z}(u_{n_{1}}v_{n_{1}}-u_{n_{2}}v_{n_{2}})].    
\end{align}
Substituting this into the summation, the factor independent of $i$ can be pulled out and disappears under the modulus:
\begin{align}
&\left| e^{-j\Phi _{n_1,n_2}}\sum_{i=-\frac{M-1}{2}}^{\frac{M-1}{2}}{e^{-ji\Delta _{s,t}}} \right|^2=\left| \frac{\sin\mathrm{(}M\Delta _{s,t}/2)}{\sin\mathrm{(}\Delta _{s,t}/2)} \right|^2\nonumber\\
&=f_M(\Delta _{s,t}).
\end{align}
Finally, regrouping the sum over all distinct IRS element pairs $(n_1, n_2)$ by their relative spatial lags $(s, t)$, and noting that the case $(s,t)=(0,0)$ corresponds to $n_1=n_2$ and is excluded by Lemma 1, yields the expression in the theorem with multiplicity weight $w(s,t)$. 
\end{proof}
\end{Theorem}
Theorem 1 establishes a direct link between the statistical user correlation and the BS array factor under random phasing. Unlike the constant correlation observed in far-field LoS channels, the appearance of the Dirichlet kernel $f_{M}(\cdot)$ implies that the random interference floor is structured and governed by the phase difference $\Delta_{s,t}$. This reveals that the inter-user interference is not an immutable environmental property, but a manageable variable determined by how the sparse array geometry samples the near-field wavefront.

\subsection{Design Principles for Channel Decorrelation }

The metric $g_{jk}$ in \eqref{gjk}, governed by the Dirichlet kernel, captures the complex interplay between the system's physical parameters. To distill fundamental design principles from this expression, we first analyze its asymptotic behavior. This reveals the distinct and crucial roles of the BS and IRS scales in enabling multi-user spatial multiplexing, leading to our first key proposition.
\begin{figure}[t]
\centering
\includegraphics[width=0.5\textwidth]{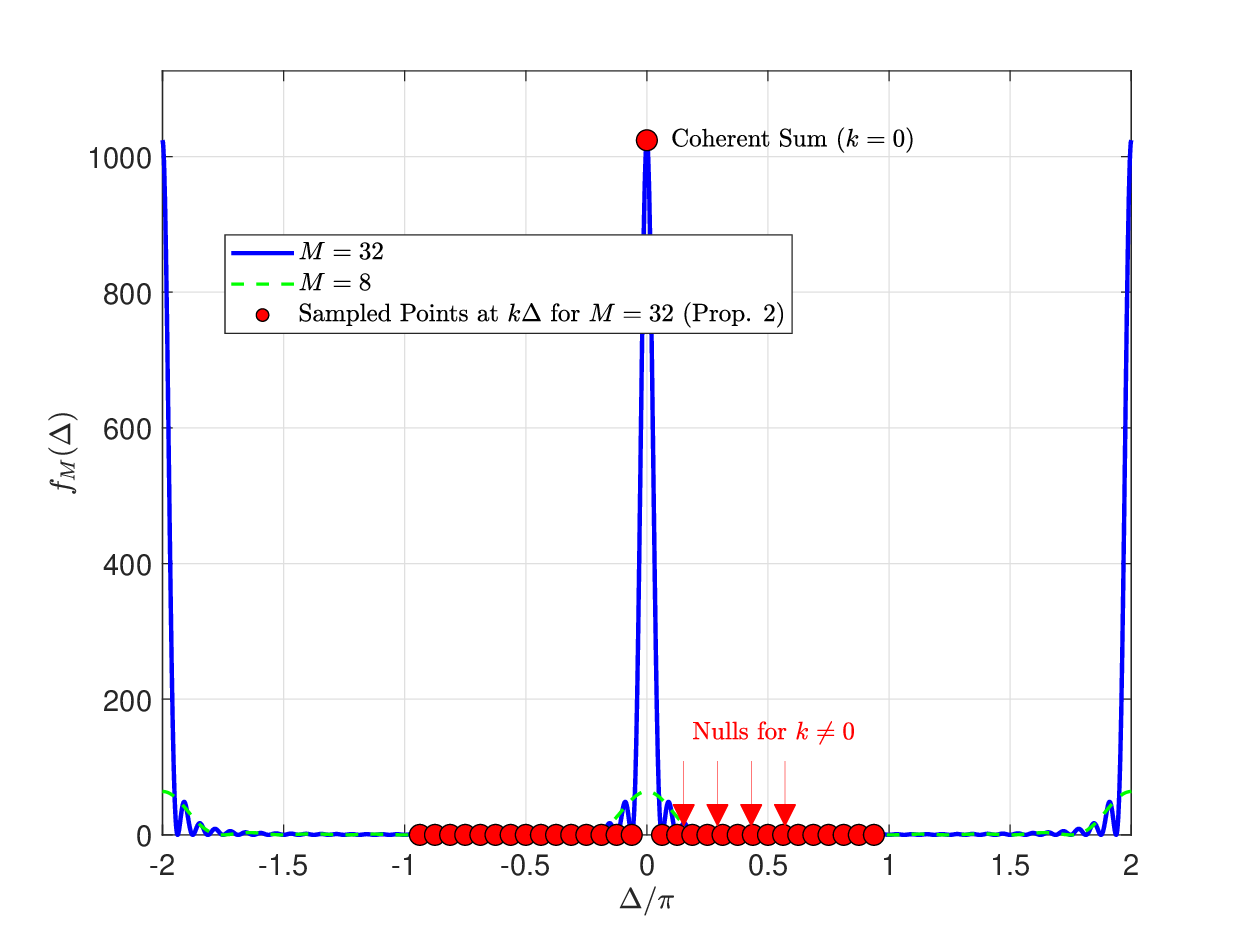}
\caption{Dirichlet kernel $f_M(\Delta_{s,t})$ and sampled nulls.}
\label{gjkFig}
\vspace{-3pt}
\end{figure}
\begin{pos}
For a system with fixed BS antennas $M$ and random IRS phase shifts, $g_{jk} \to 0$ cannot be achieved by solely increasing the IRS size $N$. In the asymptotic limit $N \to \infty$, the metric converges to a non-zero floor determined by $M$:
\begin{align}
 \lim_{N \to \infty} g_{jk} = \mathcal{C}(M) > 0,   
\end{align}
where $\mathcal{C}(M)$ is a constant independent of $N$.
\end{pos}

\begin{proof}
The behavior of $g_{jk}$ is governed by the asymptotic scaling of its numerator and denominator. The denominator, representing the product of expected channel powers, scales quadratically with the IRS size, i.e., it is of order $O(N^2)$. For the numerator in \eqref{gjk}, as $N \to \infty$, the discrete summation over the relative indices $(s,t)$ can be approximated by a Riemann integral over the normalized aperture domain. Since the Dirichlet kernel $f_M(\cdot)$ is non-negative and its main lobe energy is independent of $N$, the weighted summation also exhibits an asymptotic growth rate of $O(N^2)$. Consequently, the ratio defining $g_{jk}$ converges to a strictly positive constant $\mathcal{C}(M)$.
\end{proof}
Proposition 1 indicates that channel decorrelation is primarily enabled by a large BS aperture. This clarifies the distinct yet complementary roles of the BS and the IRS. The number of BS antennas $M$ governs spatial separability, whereas the number of IRS elements $N$ mainly contributes passive array gain to the link budget. This observation naturally leads to the question of whether we can move from merely exploiting favorable statistics to deterministically engineering decorrelated channels.

The structure of the Dirichlet kernel in \eqref{gjk} suggests that such deterministic engineering is indeed possible. As visualized in Fig. \ref{gjkFig}, the kernel exhibits a series of deep, periodic nulls. This deterministic structure presents an opportunity for proactive channel engineering: if the system geometry can be designed to place the inter-element phase differences $\Delta_{s,t}$ precisely at these nulls, the cross-correlation terms in \eqref{gjk} can be systematically eliminated. Leveraging this ability to exploit the array factor's nulls through strategic geometric design, we establish the deployment criterion in the following proposition.

\begin{pos}
Consider a symmetric IRS deployment satisfying $\mu_y = \mu_z$. Let the fundamental phase increment be defined as $\Delta \triangleq \frac{\pi\lambda}{2l}\zeta_{\text{BS}}\zeta_{\text{IRS}}\mu_x\mu_y$. If the system is designed such that 
\begin{align}
    \Delta =\frac{q\pi}{M},
\end{align}
where $q$ is an even integer not a multiple of $M$ and $\mathrm{gcd}\left( \frac{q}{2},M \right) =1$ with $M > N_{u} + N_{v} - 2$, then the metric $g_{jk}$ simplifies to
\begin{align}\label{progjk}
g_{jk} = \frac{1}{(N_uN_v)^2} \sum_{\substack{s=-(N_{\min}-1)\\ s\neq 0}}^{N_{\min}-1} (N_u-|s|)(N_v-|s|),    
\end{align}
where $N_{\min} = \min(N_u, N_v)$.
\end{pos}
\begin{proof}
Under the symmetric deployment condition $\mu_y = \mu_z$, the two-dimensional phase dependency $\Delta_{s,t}$ collapses to a one-dimensional form, $\Delta_{s,t} = (s+t)\Delta$. By defining an auxiliary index $k=s+t$, the summation for $g_{jk}$ can be regrouped as
\begin{align}
 g_{jk} \propto \sum_{k=-(N_u+N_v-2)}^{N_u+N_v-2} W(k) f_M(k\Delta),    
\end{align}
where $W(k) = \sum_{s+t=k} (N_u-|s|)(N_v-|t|)$ is the aggregated weight for a given $k$. The condition $\Delta = q\pi/M$ is specifically chosen to place the nulls of the Dirichlet kernel $f_M(\cdot)$. Substituting $\Delta = q\pi/M$ into the kernel term gives
\begin{align}
f_M(k\Delta)=\left(\frac{\sin(kq\pi/2)}{\sin(kq\pi/(2M))}\right)^2.
\end{align}
Since $q$ is an even integer, the numerator $\sin(kq\pi/2)$ is zero for every integer $k$. For $f_M(k\Delta)$ to be non-zero, the denominator must also be zero, which requires
\begin{align}
\frac{kq}{2M}\in\mathbb{Z}.
\end{align}
Using the condition $\mathrm{gcd}\left(\frac{q}{2},M\right)=1$, this is possible only if $k$ is a multiple of $M$.

The summation range for $k$ is bounded by $|k| \le (N_u-1)+(N_v-1) = N_u+N_v-2$. The premise $M>N_u+N_v-2$ thus ensures that the only integer multiple of $M$ within this range is $k=0$. Consequently, all terms in the summation vanish except for those corresponding to $k=0$.

For $k=0$, the kernel evaluates to its peak value, $f_M(0)=M^2$, and the condition $k=s+t=0$ implies $t=-s$. Since Theorem 1 excludes the case $(s,t)=(0,0)$, the term $s=0$ must be removed. The expression for $g_{jk}$ therefore simplifies to
\begin{align}
 g_{jk}= \frac{1}{(N_uN_v)^2}\sum_{\substack{s=-(N_{\min}-1)\\ s\neq 0}}^{N_{\min}-1}(N_u-|s|)(N_v-|s|).
\end{align}
The summation is over all valid indices $s$ such that $t=-s$ is also a valid index. This constrains $|s| \le N_{\min}-1$, which completes the proof.
\end{proof}

Proposition 2 captures the near-field geometric effect through the closed-form correlation term $g_{jk}$ in \eqref{progjk}. To obtain simple design rules with respect to the IRS size and the BS aperture, we next specialize to a square IRS with $N_u=N_v=\sqrt{N}$, which leads to the following corollary.

\begin{Corollary }
Under the symmetric deployment conditions of Proposition 2, for an IRS with $N_u=N_v=\sqrt{N}$, the residual channel correlation metric $g_{jk}$ is strictly upper-bounded by
\begin{align}
    g_{jk} < \frac{2}{3\sqrt{N}}.
\end{align}

\begin{proof}
When $N_u=N_v=\sqrt{N}$, $g_{jk}$ in \eqref{progjk} becomes $g_{jk} = \frac{1}{N^2} \sum_{\substack{s=-(\sqrt{N}-1)\\ s\neq 0}}^{\sqrt{N}-1} (\sqrt{N}-|s|)^2$. By symmetry, this can be rewritten as $g_{jk} = \frac{2}{N^2}\sum_{r=1}^{\sqrt{N}-1} r^2$. Applying the standard formula for the sum of squares, we obtain
\begin{align}
g_{jk}&=\frac{2}{N^2}\cdot \frac{(\sqrt{N}-1)\sqrt{N}(2\sqrt{N}-1)}{6}
\\
&=\frac{2N-3\sqrt{N}+1}{3N^{3/2}}<\frac{2}{3\sqrt{N}}.
\end{align}
This completes the proof.
\end{proof}
\end{Corollary }

Proposition 2 and Corollary 1 reveal that a large-scale sparse MIMO system can be engineered to create a favorable transmission condition through judicious geometric design. By strategically selecting the IRS placement and array sparsities to align inter-element phase differences with the nulls of the Dirichlet kernel, a high degree of channel decorrelation is achieved. This geometric approach provides a powerful mechanism to counteract the severe inter-user interference typically caused by sparse arrays, thereby rendering the channel substantially more favorable for spatial multiplexing. It thus alleviates the need for highly complex interference cancellation schemes at the transceiver and elevates deployment geometry to a powerful new degree of freedom for interference management.
\section{Ergodic sum-rate Maximization}
\subsection{Problem Formulation}
While the preceding LoS analysis unveiled fundamental physical principles, practical wireless environments are characterized by multipath scattering. To model such scenarios, we now extend our framework by considering a Rician fading model for the IRS-user links, which adeptly captures both the deterministic LoS path and random non-line-of-sight (NLoS) components. The channel for user $k$ is thus modeled as
\begin{align}
\mathbf{g}_k&=\sqrt{\frac{\alpha _k\kappa _k}{\kappa _k+1}}\bar{\mathbf{g}}_k+\sqrt{\frac{\alpha _k}{\kappa _k+1}}\tilde{\mathbf{g}}_k\\
&=\sqrt{\alpha _k\beta _k}\bar{\mathbf{g}}_k+\sqrt{\alpha _k(1-\beta _k)}\tilde{\mathbf{g}}_k,
\end{align}
where $\alpha_k$ is the large-scale path gain, $\beta _k=\frac{\kappa _k}{\kappa _k+1}$, $\kappa _k$ is the Rician factor, $\bar{\mathbf{g}}_k$ is the deterministic LoS component, and $\tilde{\mathbf{g}}_k \sim \mathcal{CN}(\mathbf{0}, \mathbf{I}_N)$ represents the normalized NLoS fading.

The favorable channel statistics designed by our deployment strategy, which significantly reduce inter-user channel correlation, motivate the adoption of a low-complexity precoding scheme. While precoders like Zero-Forcing (ZF) can eliminate interference, they entail high computational complexity due to matrix inversions, particularly in large-scale systems. Consequently, with inter-user interference already substantially suppressed by our near-field deployment, we employ MRT to maximize the desired signal power with reduced complexity. The precoding vector at the BS for user $k$ under the MRT scheme is $\mathbf{w}_k = \sqrt{p_k}\mathbf{h}_k$, where $p_k$ is the power allocation factor for user $k$. Furthermore, $p_k$ is optimized under the total average transmit power constraint $\sum_{k=1}^K p_k \mathbb{E}[\|\mathbf{h}_k\|^2] \le P_{\max}$, where the average is taken over the Rician fading statistics \cite{r27}.
 It is worth noting that the proposed design does not require full online estimation of the large-dimensional cascaded BS-IRS-user channel in every coherence interval. Specifically, the IRS phase shifts are optimized based on slowly varying statistical CSI, while the BS beamforming is updated according to the effective instantaneous CSI under the fixed BS-IRS deployment. Our objective is to maximize the long-term system performance, which is characterized by the ergodic sum-rate as follows:
\begin{align}\label{Rerg}
R_{\text{erg}} = \sum_{k=1}^K \mathbb{E}_{\{\tilde{\mathbf{g}}_k\}} \left[ \log_2(1 + \gamma_k) \right]. 
\end{align}
However, directly optimizing \eqref{Rerg} is difficult and a widely adopted and effective method to handle it is to optimize an approximation of the ergodic sum-rate based on statistical CSI. This leads to an effective SINR expression for user $k$:
\begin{align}\label{E1}
\gamma_k^{\mathrm{approx}} \triangleq \frac{p_k|\mathbb{E}[\mathbf{h}_k^H\mathbf{h}_k]|^2}{\sum_{j \neq k} p_j \mathbb{E}[|\mathbf{h}_k^H\mathbf{h}_j|^2] + \sigma_k^2},    
\end{align}
where
\begin{align}\label{E2}
\mathbb{E}[\mathbf{h}_k^H\mathbf{h}_k] &= C_k + \alpha_k\beta_k \boldsymbol{\theta}^H \mathbf{A}_{k,k} \boldsymbol{\theta}, \\
\mathbb{E}[|\mathbf{h}_k^H\mathbf{h}_j|^2] &\!\!=\!\! \alpha_k\alpha_j\beta_k\beta_j |\boldsymbol{\theta}^H \mathbf{A}_{k,j} \boldsymbol{\theta}|^2 \!+\! \alpha_k\alpha_j\beta_k(1-\beta_j) \boldsymbol{\theta}^H \mathbf{D}_k \boldsymbol{\theta} \nonumber \\
& \quad \!+\! \alpha_k\alpha_j\beta_j(1-\beta_k) \boldsymbol{\theta}^H \mathbf{D}_j \boldsymbol{\theta} + C_{k,j}
\end{align}
with $C_k=\alpha _k(1-\beta _k)\|\mathbf{F}\|_{F}^{2}$, $C_{k,j}=\alpha _k\alpha _j(1-\beta _k)(1-\beta _j)\mathrm{tr((}\mathbf{FF}^H)^2)
$, $\mathbf{A}_{k,j}=\mathrm{diag}(\bar{\mathbf{g}}_k)^H\mathbf{F}^H\mathbf{F}\mathrm{diag}(\bar{\mathbf{g}}_j)$, $\mathbf{D}_k = \mathrm{diag}(\bar{\mathbf{g}}_k)^H \mathbf{F}^H \mathbf{F}\mathbf{F}^H \mathbf{F} \mathrm{diag}(\bar{\mathbf{g}}_k)$, $k,j\in \left\{ 1,2,…,K \right\}$ and $\boldsymbol{\theta }=\left[ e^{j\phi _1},e^{j\phi _2},\cdots ,e^{j\phi _N} \right] ^T$.
Please refer to Appendix B for more details.

Consequently, our optimization problem transforms into maximizing the approximate ergodic sum-rate by jointly designing the IRS phase shifts and the power allocation. The problem is thus formulated as
\begin{align} \label{P_Ergodic_Main}
   \underset{\boldsymbol{\theta}, \mathbf{p}}{\max} \quad & \sum_{k=1}^K \log_2 \left(1 + \gamma_k^{\mathrm{approx}}\right) \\
    \text{s.t.} \quad & |[\boldsymbol{\theta}]_n| = 1,  \forall n \in \{1, \dots, N\}, \label{Theta Constraints}\tag{\ref{P_Ergodic_Main}{a}} \\
    & \sum_{k=1}^K p_k \mathbb{E}[\|\mathbf{h}_k\|^2] \le P_{\max},  p_k \ge 0,  \forall k, \label{Power Constraints}\tag{\ref{P_Ergodic_Main}{b}} 
\end{align}
where $\mathbf{p} = [p_1, \dots, p_K]^T$.

Since problem \eqref{P_Ergodic_Main} is highly non-convex, to overcome it, we first reformulate the objective function into a more tractable form using a sequence of transformations, which enables an efficient alternating optimization algorithm.
First, using the Lagrangian dual method, we introduce auxiliary variables $\boldsymbol{\xi} = [\xi_1, \dots, \xi_K]^T$\cite{r25}. The problem is equivalent to maximizing
\begin{align}\label{L1}
\!\!\mathcal{L} _1(\boldsymbol{\theta },\mathbf{p},\boldsymbol{\xi })\!=\!\sum_{k=1}^K{\left[ \log _2(1+\xi_k)\!-\!\xi_k\!+\!\frac{(1+\xi_k)S_k(\boldsymbol{\theta },p_k)}{J_k(\boldsymbol{\theta },\mathbf{p})\!+\!S_k(\boldsymbol{\theta },p_k)} \right]},
\end{align}
where $S_k(\boldsymbol{\theta}, p_k) = p_k|\mathbb{E}[\mathbf{h}_k^H\mathbf{h}_k]|^2$ represents the statistical signal power and $J_k(\boldsymbol{\theta}, \mathbf{p}) = \sum_{j \neq k} p_j \mathbb{E}[|\mathbf{h}_k^H\mathbf{h}_j|^2] + \sigma_k^2$ represents the statistical interference-plus-noise power. For fixed $(\boldsymbol{\theta}, \mathbf{p})$, the closed-form solution for \eqref{L1} is 
\begin{align}\label{Optimizex}
\xi _{k}^{\star}=\frac{S_k(\boldsymbol{\theta },p_k)}{J_k(\boldsymbol{\theta },\mathbf{p})},k=1,…,K.
\end{align}

Next, we apply the quadratic transform with a second set of auxiliary variables $\boldsymbol{\eta} = [\eta_1, \dots, \eta_K]^T$\cite{r26}. This results in the objective:
\begin{align}
\mathcal{L} _2(\boldsymbol{\theta },\mathbf{p},\boldsymbol{\xi },\boldsymbol{\eta })=\sum_{k=1}^K{\left[ 2\eta _k\sqrt{(1+\xi _k)S_k(\boldsymbol{\theta },p_k)} \right.}
\nonumber\\
-\eta _{k}^{2}\left( J_k\left( \boldsymbol{\theta },\mathbf{p} \right) +S_k(\boldsymbol{\theta },p_k) \right) ].
\end{align}
For fixed $(\boldsymbol{\theta}, \mathbf{p}, \boldsymbol{\xi})$, the closed-form solution is 
\begin{align}\label{Optimizey}
\eta _{k}^{\star}=\frac{\sqrt{(1+\xi _k)S_k(\boldsymbol{\theta },p_k)}}{S_k(\boldsymbol{\theta },p_k)+J_k(\boldsymbol{\theta },\mathbf{p})},k=1,…,K.
\end{align}

The algorithm proceeds by iteratively updating $(\boldsymbol{\xi}, \boldsymbol{\eta})$, $\mathbf{p}$, and $\boldsymbol{\theta}$.

\subsection{Power Allocation Optimization}
With the IRS phases $\boldsymbol{\theta}$ and auxiliary variables $(\boldsymbol{\xi}, \boldsymbol{\eta})$ held fixed, the subproblem for optimizing the power allocation vector $\mathbf{p}$ is derived from maximizing $\mathcal{L}_2$. Isolating the terms dependent on $p_k$ for a single user, the objective for each $p_k$ becomes
\begin{align} \label{P_Power}
   \underset{p_k}{\max} \quad & -c_{1,k} p_k + c_{2,k} \sqrt{p_k} \\
    \text{s.t.} \quad & \eqref{Power Constraints}. \nonumber
\end{align}
where $c_{1,k} = \eta_k^2 |\mathbb{E}[\mathbf{h}_k^H\mathbf{h}_k]|^2 + \sum_{j \neq k} \eta_j^2 \mathbb{E}[|\mathbf{h}_j^H\mathbf{h}_k|^2]$ and $c_{2,k} = 2\eta_k \sqrt{1+\xi_k} |\mathbb{E}[\mathbf{h}_k^H\mathbf{h}_k]|$. Problem \eqref{P_Power} is convex for each $p_k$. Thus, we can solve for the entire vector $\mathbf{p}$ using a block coordinate descent (BCD) approach, where each $p_k$ is updated iteratively. The optimal solution for a single $p_k$, given the others, has a closed-form as follows:
\begin{align}\label{OptimizePower}
p_k^{\star} = \min \left\{ \frac{P_{\max} - \sum_{j \neq k} p_j \mathbb{E}[\|\mathbf{h}_j\|^2]}{\mathbb{E}[\|\mathbf{h}_k\|^2]}, \left( \frac{c_{2,k}}{2c_{1,k}} \right)^2 \right\}.
\end{align}

\subsection{IRS Phase Shift Optimization}
With the power allocation $\mathbf{p}$ and auxiliary variables $(\boldsymbol{\xi}, \boldsymbol{\eta})$ fixed, the remaining subproblem is to optimize the IRS phase vector $\boldsymbol{\theta}$. Based on the surrogate function $\mathcal{L}_2$, it is equivalent to the following minimization problem:
\begin{align} \label{P_Theta}
   \underset{\boldsymbol{\theta}}{\min} \quad & \boldsymbol{\theta}^H \mathbf{Q} \boldsymbol{\theta} + \sum_{k=1}^K \sum_{j=1}^K p_j \eta_k^2 \alpha_k\alpha_j\beta_k\beta_j |\boldsymbol{\theta}^H \mathbf{A}_{k,j} \boldsymbol{\theta}|^2 \\
    \text{s.t.} & \quad |[\boldsymbol{\theta }]_n|\leqslant 1,\forall n\in \{1,\dots ,N\},\label{ThetaNewConstraint}\tag{\ref{P_Theta}{a}} \\
    &\quad\eqref{Power Constraints}, \nonumber
\end{align}
where the complex matrix $\mathbf{Q}$ is defined as
\begin{align}
&\mathbf{Q} = -\left[ \sum_{k=1}^K \left( 2\eta_k \alpha_k\beta_k \sqrt{(1+\xi_k)p_k} - 2\eta_k^2 p_k \alpha_k C_k \beta_k \right) \mathbf{A}_{k,k} \right] \nonumber\\
& \!+\!\sum_{k=1}^K{\eta _{k}^{2}}\sum_{j\ne k}^K{p_j}\alpha _k\alpha _j\left( \beta _k(1-\beta _j)\mathbf{D}_k\!+\!\beta _j(1-\beta _k)\mathbf{D}_j \right).
\end{align}

To address the non-convexity of the unit-modulus constraint in \eqref{Theta Constraints}, we first relax it to the convex inequality in \eqref{ThetaNewConstraint} for tractable optimization, and then project the converged solution back via the mapping $\theta_n \leftarrow e^{j\arg\{\theta_n\}}$. However, problem \eqref{P_Theta} is still highly non-convex. We solve it using the Alternating Direction Method of Multipliers (ADMM) by introducing auxiliary variables $\mathbf{z}, \mathbf{k}, \mathbf{s}$ and the consensus constraints $\mathbf{z}=\boldsymbol{\theta}, \mathbf{k}=\boldsymbol{\theta}, \mathbf{s}=\boldsymbol{\theta}$. The augmented Lagrangian is formulated as
\begin{align} \label{Lrho}
   \underset{\boldsymbol{\theta},\mathbf{z}, \mathbf{k}, \mathbf{s}}{\min} \quad & \mathcal{L}_{\rho}(\boldsymbol{\theta}, \mathbf{z}, \mathbf{k}, \mathbf{s}) \\
    \text{s.t.} & \quad \mathbf{z}=\boldsymbol{\theta}, \mathbf{k}=\boldsymbol{\theta}, \mathbf{s}=\boldsymbol{\theta},\tag{\ref{Lrho}{a}} \\ 
    & \quad |[\mathbf{k}]_n|\leqslant 1, \forall n,\tag{\ref{Lrho}{b}}\label{Kconstrains}\\
    & \quad \sum_{k=1}^K{p_k}\left( C_k+\alpha _k\beta _k\mathbf{s}^H\mathbf{A}_{k,k}\mathbf{s} \right) \le P_{\max},\label{Sconstraints}\tag{\ref{Lrho}{c}}
\end{align}
where $\mathcal{L} _{\rho}(\boldsymbol{\theta },\mathbf{z},\mathbf{k},\mathbf{s})=\Re \{\boldsymbol{\theta }^H\mathbf{Qz}\}+\rho \parallel \mathbf{z}\!-\!\boldsymbol{\theta }+\boldsymbol{\lambda }_z\parallel _{2}^{2}+\sum_{k,j}{p_j}\eta _{k}^{2}\alpha _k\alpha _j\beta _k\beta _j|\boldsymbol{\theta }^H\mathbf{A}_{k,j}\mathbf{z}|^2\!+\!\rho \parallel \mathbf{k}-\boldsymbol{\theta }+\boldsymbol{\lambda }_k\parallel _{2}^{2}\!+\!\rho \parallel \mathbf{s}-\boldsymbol{\theta }+\boldsymbol{\lambda }_s\parallel _{2}^{2}$, $\rho>0$ is the penalty parameter, and $\boldsymbol{\lambda}_z, \boldsymbol{\lambda}_k, \boldsymbol{\lambda}_s$ are the scaled dual variables. The variables are updated iteratively:
\subsubsection{Update for $\boldsymbol{\theta}$} With $(\mathbf{z}, \mathbf{k}, \mathbf{s})$ fixed, the subproblem for $\boldsymbol{\theta}$ is an unconstrained quadratic minimization:
\begin{align}
\underset{\boldsymbol{\theta}}{\min} \quad \boldsymbol{\theta}^H(\mathbf{\Upsilon}(\mathbf{z}) + 3\rho\mathbf{I})\boldsymbol{\theta} - 2\Re\left\{\boldsymbol{\theta}^H \mathbf{v}_{\theta} \right\}
\end{align}
    where $\mathbf{\Upsilon}(\mathbf{z}) = \sum_{k,j} \eta_k^2 p_j \alpha_k\alpha_j\beta_k\beta_j (\mathbf{A}_{k,j}\mathbf{z})(\mathbf{z}^H\mathbf{A}_{k,j}^H)$ and we define the composite vector $\mathbf{v}_{\theta}$ for notational simplicity:
\begin{align}
\mathbf{v}_{\theta} \triangleq \rho(\mathbf{z}+\boldsymbol{\lambda}_z + \mathbf{k}+\boldsymbol{\lambda}_k + \mathbf{s}+\boldsymbol{\lambda}_s) - \frac{1}{2}\mathbf{Qz}.
\end{align}
The closed-form solution is found by setting the gradient to zero:
\begin{align}\label{OptimizeTheta}
\boldsymbol{\theta}^\star = (\mathbf{\Upsilon}(\mathbf{z}) + 3\rho\mathbf{I})^{-1} \mathbf{v}_{\theta}.
\end{align}

\subsubsection{Update for $\mathbf{z}$} Similarly, with other variables fixed, the subproblem for $\mathbf{z}$ is:
\begin{align}\label{ZOpti}
\underset{\mathbf{z}}{\min} \quad \mathbf{z}^H(\mathbf{\Upsilon}'(\boldsymbol{\theta}) + \rho\mathbf{I})\mathbf{z} - 2\Re\left\{\mathbf{z}^H \mathbf{v}_{z} \right\}    
\end{align}
where $\mathbf{\Upsilon}'(\boldsymbol{\theta}) = \sum_{k,j} \eta_k^2 p_j \alpha_k\alpha_j\beta_k\beta_j (\mathbf{A}_{k,j}^H\boldsymbol{\theta})(\boldsymbol{\theta}^H\mathbf{A}_{k,j})$ and the composite vector $\mathbf{v}_{z}$ is
\begin{align}
\mathbf{v}_{z} \triangleq \rho(\boldsymbol{\theta}-\boldsymbol{\lambda}_z) - \frac{1}{2}\mathbf{Q}^H\boldsymbol{\theta}. 
\end{align}
The closed-form solution is
\begin{align}\label{OptimizeZ}
\mathbf{z}^\star = (\mathbf{\Upsilon}'(\boldsymbol{\theta}) + \rho\mathbf{I})^{-1} \mathbf{v}_{z}.   
\end{align}

\subsubsection{Update for $\mathbf{k}$} The subproblem for $\mathbf{k}$ the Euclidean projection onto the relaxed unit-disk feasible set:
    \begin{align}
       \underset{\mathbf{k}}{\min}\quad &\parallel \mathbf{k}-(\boldsymbol{\theta }-\boldsymbol{\lambda }_k)\parallel _{2}^{2} \\
        \text{s.t.} \quad & \eqref{Kconstrains}. \nonumber
    \end{align}
The optimal update for each element of $\mathbf{k}$ is given in the following closed form: 
\begin{align}\label{optimizeK}
[\mathbf{k}^{\star}]_n=\begin{cases}
	[\boldsymbol{\theta }-\boldsymbol{\lambda }_k]_n,\mathrm{if} |[\boldsymbol{\theta }-\boldsymbol{\lambda }_k]_n|\le 1.\\
	\frac{[\boldsymbol{\theta }-\boldsymbol{\lambda }_k]_n}{|[\boldsymbol{\theta }-\boldsymbol{\lambda }_k]_n|},\mathrm{otherwise}.
\end{cases}
    \end{align}

\begin{algorithm}[t]
\caption{Proposed Algorithm for Problem\eqref{P_Ergodic_Main}}
\label{alg:main_algorithm}
\begin{algorithmic}[1]
\STATE \textbf{Initialize:} Set iteration counter $t=0$, tolerance $\delta$. Initialize variables $\mathbf{p}^{(0)}$, $\boldsymbol{\theta}^{(0)}$, $\mathbf{z}^{(0)}$, $\mathbf{k}^{(0)}$, $\mathbf{s}^{(0)}$ and dual variables $\boldsymbol{\lambda}_z^{(0)}, \boldsymbol{\lambda}_k^{(0)}, \boldsymbol{\lambda}_s^{(0)}$ to feasible values.

\STATE \textbf{repeat}
\STATE Update auxiliary variables $\boldsymbol{\xi}^{(t+1)}$ and $\boldsymbol{\eta}^{(t+1)}$ by \eqref{Optimizex} and \eqref{Optimizey}.
\STATE Update power allocation $\mathbf{p}^{(t+1)}$ by \eqref{OptimizePower}.
\STATE Update primal variables $\boldsymbol{\theta}^{(t+1)}$, $\mathbf{z}^{(t+1)}$, $\mathbf{k}^{(t+1)}$, $\mathbf{s}^{(t+1)}$ by \eqref{OptimizeTheta}, \eqref{OptimizeZ}, \eqref{optimizeK}, and \eqref{OptimizeS}, respectively.
\STATE Update dual variables $\boldsymbol{\lambda}_z^{(t+1)}, \boldsymbol{\lambda}_k^{(t+1)}, \boldsymbol{\lambda}_s^{(t+1)}$ by \eqref{dualUpdate}.
\STATE Set $t = t+1$.
\STATE \textbf{until} the improvement of the objective value of \eqref{P_Ergodic_Main} is below $\delta$.
\STATE \textbf{Output:} The optimized power allocation $\mathbf{p}^{\star} = \mathbf{p}^{(t)}$ and phase shifts $\theta_n \leftarrow e^{j\arg\{\theta_n\}}$.
\end{algorithmic}
\end{algorithm}

\subsubsection{Update for $\mathbf{s}$} The subproblem for $\mathbf{s}$ is the convex quadratically constrained quadratic program (QCQP):
    \begin{align}\label{OptimizeS}
        \underset{\mathbf{s}}{\min} \quad & \|\mathbf{s}-(\boldsymbol{\theta}-\boldsymbol{\lambda}_s)\|^2 \\
        \text{s.t.} \quad & \eqref{Sconstraints}. \nonumber
    \end{align}
Since \eqref{OptimizeS} involves only a single quadratic power constraint, it can be solved efficiently by a one-dimensional bisection over the associated Lagrange multiplier, or alternatively by standard convex optimization solvers such as CVX \cite{r28}.

\subsubsection{Dual Variable Updates} The dual variables are updated via standard gradient ascent steps as follows:
\begin{align}\label{dualUpdate}
\begin{cases}
	\boldsymbol{\lambda }_{z}^{t+1}=\boldsymbol{\lambda }_{z}^{t}+\mathbf{z}-\boldsymbol{\theta }\\
	\boldsymbol{\lambda }_{k}^{t+1}=\boldsymbol{\lambda }_{k}^{t}+\mathbf{k}-\boldsymbol{\theta }\\
	\boldsymbol{\lambda }_{s}^{t+1}=\boldsymbol{\lambda }_{s}^{t}+\mathbf{s}-\boldsymbol{\theta }\\
\end{cases}
\end{align}
where $t$ represents the iteration count.

The overall algorithm for \eqref{P_Ergodic_Main} is summarized in Algorithm 1. The alternating optimization structure guarantees the algorithm's convergence to a stationary point.  The computational cost of Algorithm 1 can be divided into an offline preprocessing stage and an online iterative stage. In the preprocessing stage, forming the statistical matrices related to $\mathbf{F}^{H}\mathbf{F}$, $\{\mathbf{A}_{k,j}\}$, and $\{\mathbf{D}_{k}\}$ requires complexity $\mathcal{O}(MN^{2} + N^{3} + K^{2}N^{2})$. In each outer iteration, the updates of the auxiliary variables require $\mathcal{O}(K^{2}N^{2})$, while the power-allocation update admits a closed-form solution with much lower complexity. The dominant online cost comes from the IRS phase-shift optimization, whose ADMM/QCQP step requires $\mathcal{O}(N^{3.5})$. Therefore, the overall complexity of Algorithm 1 is
$\mathcal{O}\big(MN^{2} + N^{3} + K^{2}N^{2} + I(K^{2}N^{2} + N^{3.5})\big)$, where $I$ denotes the number of alternating iterations\cite{10915665}. This result also reveals the scalability of the proposed algorithm: the complexity increases quadratically with the number of users $K$, depends mainly linearly on the number of BS antennas $M$ through the preprocessing stage, and is more sensitive to the IRS element number $N$ due to the IRS phase-shift optimization. Moreover, since the proposed design is based on long-term statistical CSI and MRT precoding, the optimization is performed on a slower timescale and does not need to be updated in every coherence block, which alleviates the practical latency burden.

\section{Numerical Results}
 In this section, we provide numerical results to validate our theoretical analyses and demonstrate the performance of the proposed optimization algorithm. Unless specified otherwise, the four single-antenna users are located at $[10, 70, 0]^T$ m, $[30, 60, 0]^T$ m, $[20, 50, 0]^T$ m, and $[45, 45, 0]^T$ m, respectively. Other parameters are given in Table~\ref{tab:sim_params}.

\begin{table}[t]
\caption{ Simulation Parameters}
\label{tab:sim_params}
\centering
\em
\scriptsize
\renewcommand{\arraystretch}{1.12}
\setlength{\tabcolsep}{3pt}
\resizebox{\columnwidth}{!}{%
\begin{tabular}{|c|l|c|}
\hline
\textbf{Parameter} & \textbf{Description} & \textbf{Value} \\
\hline
$f$ & Carrier frequency & {\centering 60 GHz\par} \\
\hline
$M$ & Number of BS antennas & {\centering 128\par} \\
\hline
$\mathbf{r}_{\mathrm{BS}}$ & BS location & $[0,\,0,\,0]^T$ m \\
\hline
$\zeta_{\mathrm{BS}}$ & BS sparsity factor & {\centering 3\par} \\
\hline
$\zeta_{\mathrm{IRS}}$ & IRS sparsity factor & {\centering 6\par} \\
\hline
$\mathbf{r}_{\mathrm{IRS}}$ & Default IRS center location & $[-0.72,\,0.51,\,0.51]^T$ m \\
\hline
$K$ & Number of users & {\centering 4\par} \\
\hline
$\sigma_k^2$ & Noise power & {\centering $-115$ dBm\par} \\
\hline
$N_{\mathrm{MC}}$ & Monte Carlo runs per point & {\centering 1000\par} \\
\hline
$\delta $ & Convergence tolerance &{\centering $10^{-3}$\par}\\
\hline
\end{tabular}
}
\vspace{-10pt}
\end{table}

Fig. \ref{varFig} illustrates the favorable propagation metric $g_{jk}$ versus the IRS's y-coordinate for various array sparsity configurations when $N=1600$ and the IRS is centered at $[-3m, l_y, 3m]$. The analytical model from \eqref{gjk} shows excellent agreement with the Monte Carlo simulations, which confirms the fidelity of our near-field channel correlation model. It also demonstrates the profound impact of array sparsity on channel decorrelation. Specifically, as the sparsity factors ($\zeta_{\mathrm{BS}}$, $\zeta_{\mathrm{IRS}}$) increase, the value of $g_{jk}$ is substantially reduced, indicating that the user channels become increasingly orthogonal. This behavior is attributed to the larger effective apertures, which accentuate the near-field wavefront curvature and thus generate more distinct phase variations across the IRS elements—a phenomenon accurately captured by the $\Delta_{s,t}$ term in our model. In contrast, the conventional dense array setup ($\zeta_{\mathrm{BS}}=1, \zeta_{\mathrm{IRS}}=1$) yields a persistently high correlation ($g_{jk}>0.75$), highlighting the necessity of sparse arrays for achieving favorable propagation in the near-field.

\begin{figure}[t]
\centering
\includegraphics[width=0.5\textwidth]{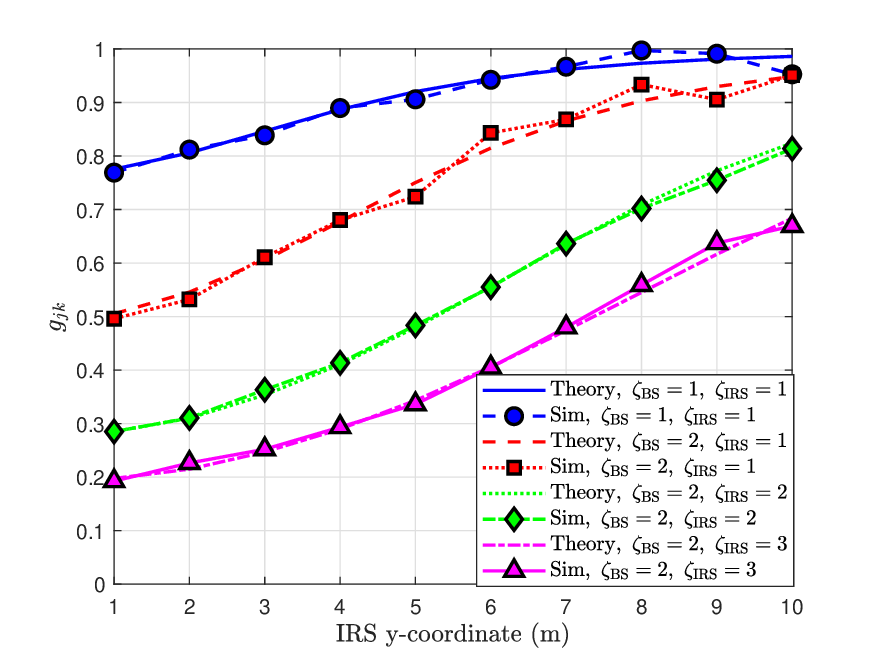}
\caption{ $g_{jk}$ versus sparsity and IRS position.}
\label{varFig}
\vspace{-3pt}
\end{figure}

 To evaluate the sensitivity of the proposed deployment rule to practical installation mismatch, we perturb the fundamental phase increment around the ideal design point in Proposition 2 and set $\Delta=\eta\Delta_0$, where $\Delta_0=q\pi/M$ denotes the ideal value. Fig. \ref{RobustFig} shows the resulting favorable-propagation metric $g_{jk}$ versus $\eta=\Delta/\Delta_0$. It is observed that the minimum is attained at $\eta=1$, as expected. More importantly, $g_{jk}$ varies smoothly around the optimum, and moderate mismatch only causes limited degradation. In particular, for perturbations around $\pm 10\%$, the resulting $g_{jk}$ remains close to the ideal analytical level and stays near the upper bound characterized by Corollary~1, indicating that the decorrelation effect predicted by Proposition 2 is robust to moderate geometric inaccuracies. Even when the mismatch increases to approximately $\pm 15\%$, the increase in $g_{jk}$ is still limited. This confirms that the proposed deployment rule is not a fragile point solution, but a practically useful design guideline with non-negligible tolerance.

 \begin{figure}[t]
\centering
\includegraphics[width=0.5\textwidth]{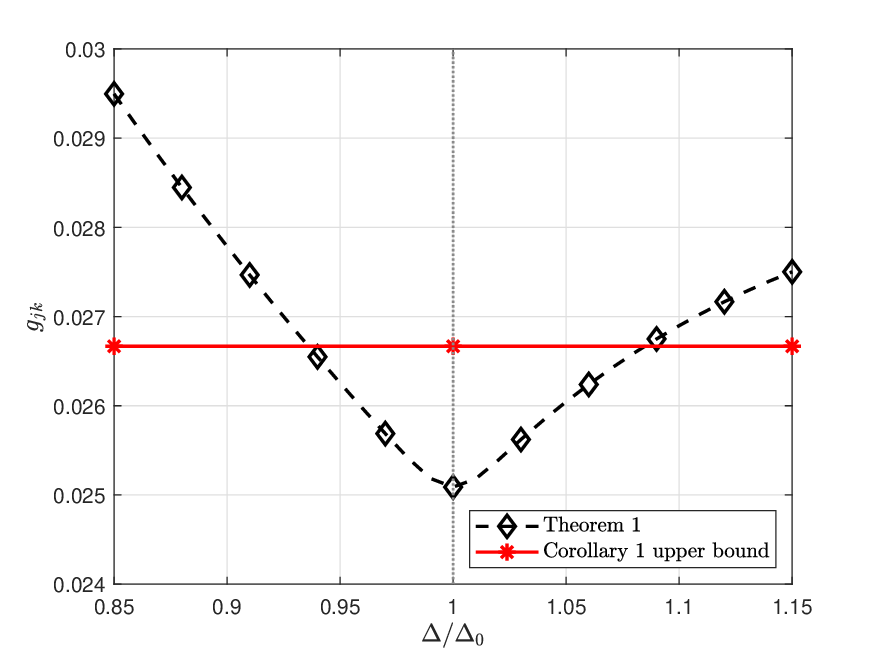}
\caption{ Robustness of the favorable-propagation metric $g_{jk}$ to phase-increment mismatch around the ideal deployment point $\Delta_0$.}
\label{RobustFig}
\vspace{-3pt}
\end{figure}

\begin{figure}[t]
\centering
\includegraphics[width=0.5\textwidth]{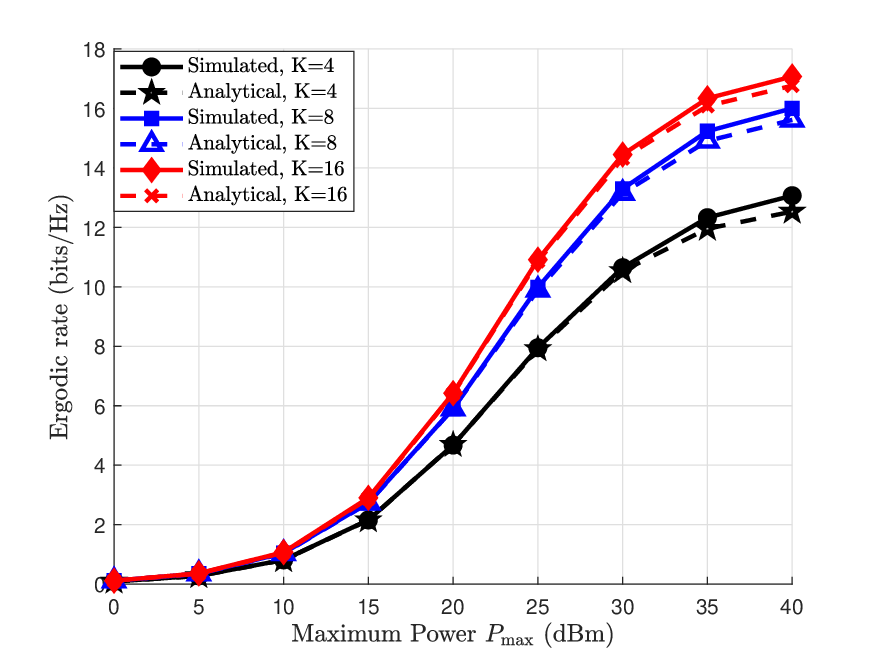}
\caption{Ergodic sum-rate versus maximum transmit power $P_{\max}$ when $N$=400.}
\label{RAppro}
\vspace{-3pt}
\end{figure}

\begin{figure}[t]
\centering
\includegraphics[width=0.5\textwidth]{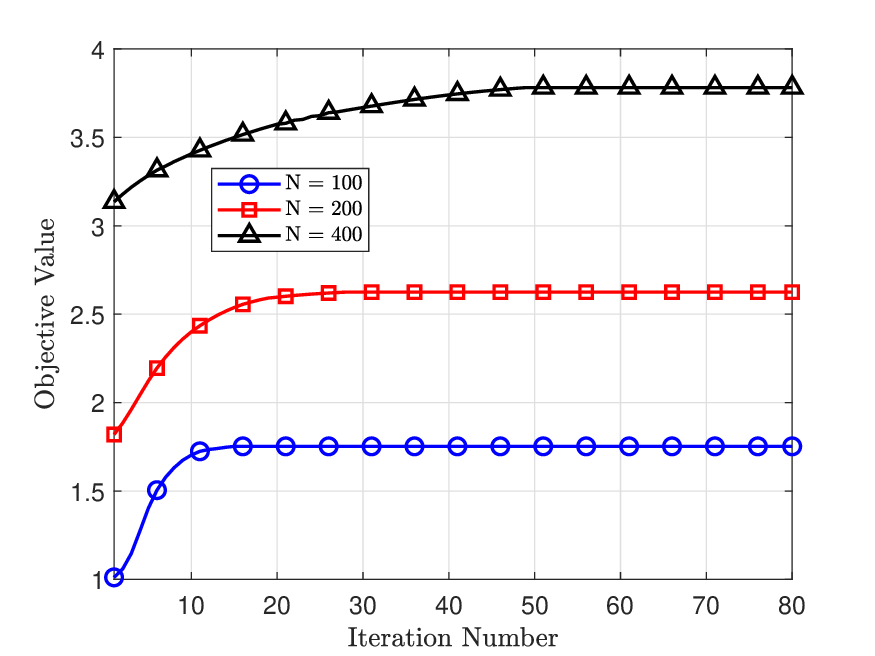}
\caption{ Convergence behavior of the proposed algorithm.}
\label{ConvergeFig}
\vspace{-3pt}
\end{figure}

\begin{figure}[t]
\centering
\includegraphics[width=0.5\textwidth]{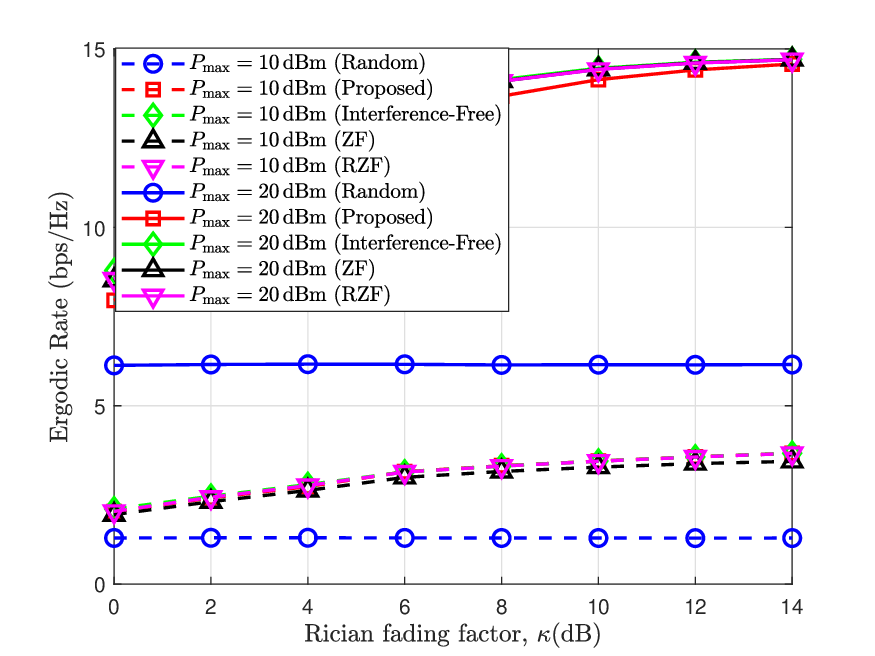}
\caption{ Ergodic sum-rate versus Rician factor $\kappa$ under different transmit power when $N$=400.}
\label{RATEFig}
\vspace{-3pt}
\end{figure}

\begin{figure}[t]
\centering
\includegraphics[width=0.5\textwidth]{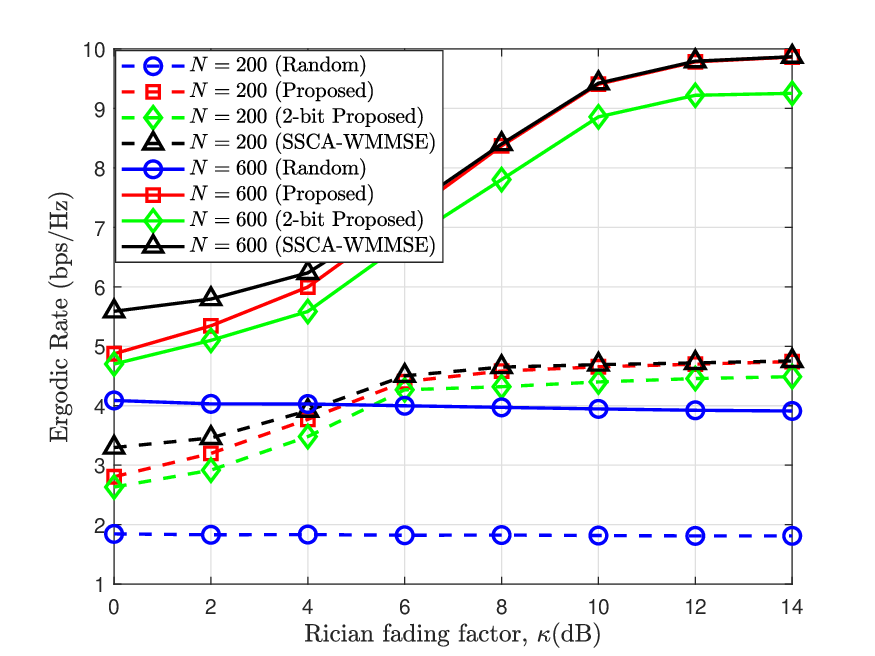}
\caption{ Ergodic sum-rate versus Rician factor $\kappa$ for different IRS elements when $P_{\max}$=15dBm.}
\label{RATEv2}
\vspace{-3pt}
\end{figure}

\begin{figure}[t]
\centering
\includegraphics[width=0.5\textwidth]{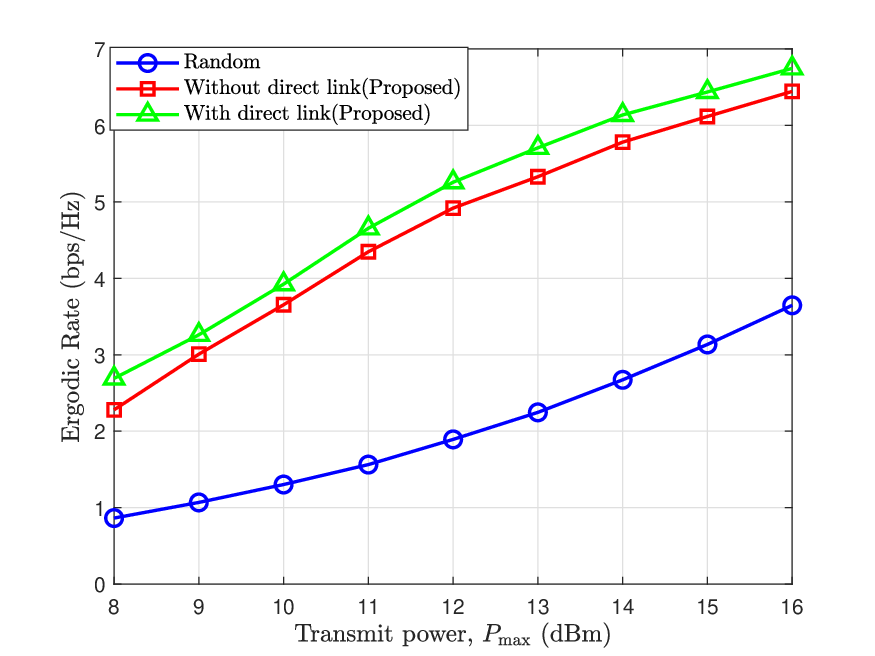}
\caption{ Effect of direct link on the ergodic rate when $N$=500.}
\label{EDoFv2}
\vspace{-3pt}
\end{figure}
Fig. \ref{RAppro} plots the ergodic sum-rate versus the maximum transmit power $P_{\max}$ to validate the accuracy of our statistical CSI-based approximation used in problem \eqref{P_Ergodic_Main}. The analytical results are benchmarked against extensive Monte Carlo simulations for systems with $K=4, 8$, and $16$ users. An excellent agreement is evident between the analytical and simulated curves for all user loads across the entire power range. This confirms that our approximation, which relies on first and second-order channel statistics, accurately captures the system's long-term performance. Furthermore, the results exhibit the expected monotonic increase in sum-rate with both $P_{\max}$ and the number of users, which aligns with fundamental system behaviors. The validated accuracy of this approximation justifies its use as a tractable yet reliable objective for the joint optimization of IRS phase shifts and power allocation.

Then, we validate the proposed algorithm's convergence. Fig. \ref{ConvergeFig} shows that the approximate ergodic sum-rate increases monotonically and converges for different numbers of IRS elements $N$. This confirms the reliable convergence of our approach. As expected, the results also demonstrate that a larger $N$ yields substantial performance gains.

 Fig. \ref{RATEFig} investigates the impact of the Rician factor $\kappa$ on the ergodic sum-rate performance of the proposed algorithm under different transmit power levels. We benchmark the performance against four schemes, including (i) Random, where the IRS phase shifts are randomly selected, (ii) Interference-Free, which serves as a theoretical benchmark, corresponding to the sum-rate achieved with the optimized IRS phase shifts from the proposed algorithm when the inter-user interference is artificially removed, (iii) ZF, where the BS employs zero-forcing precoding as a linear interference-suppression baseline, and (iv) RZF (regularized ZF), where the BS adopts regularized zero-forcing precoding to balance interference suppression and noise enhancement. As shown in Fig. \ref{RATEFig}, the proposed design consistently outperforms the Random baseline over the entire range of $\kappa$ at both transmit power levels $P_{\max}=10$ dBm and $P_{\max}=20$ dBm. In addition, the proposed scheme performs close to the Interference-Free benchmark and remains competitive with the ZF and RZF baselines, which demonstrates that the proposed near-field deployment together with the deterministic IRS phase optimization can effectively suppress inter-user interference even under the low-complexity MRT framework. It is also observed that the ergodic sum-rate of all schemes increases with $\kappa$, since a stronger LoS component renders the channel more deterministic and less sensitive to fading.

 As shown in Fig. \ref{RATEv2}, the proposed scheme achieves performance close to that of the SSCA-WMMSE benchmark\cite{zhaominmin} over the whole $\kappa$ range, while the 2-bit proposed design remains close to the continuous-phase case, indicating that most of the performance gain can be retained even with low-resolution phase control. It is worth noting that the SSCA-WMMSE benchmark relies on a stronger two-timescale joint optimization framework, where the IRS is optimized via SSCA in the long timescale and the transmit precoder is refined via WMMSE in the short timescale, and thus has a slight performance advantage. In contrast, the proposed method only uses low-complexity MRT together with statistical-CSI-based IRS phase/power optimization, with several subproblems admitting closed-form updates and without requiring heavy online iterations in every coherence block, which makes it more attractive in terms of complexity, latency, and practical implementation.

 To further assess the robustness of the proposed framework, Fig. \ref{EDoFv2} illustrates the ergodic sum-rate versus the transmit power $P_{\max}$ for the cases with and without the direct link, where the Random scheme is also included as a baseline. It is first noted that the direct BS-user link is modeled as a Rician channel with the channel-quality factor $I_{\text{direct}}=0.01$\cite{ShiJin-Channel}. Under this setting, the curve with the direct link is consistently higher than that without the direct link over the whole transmit-power range. Nevertheless, the performance improvement is only marginal. This indicates that, although the direct path provides an additional signal contribution, its strength remains limited. In contrast, the optimized IRS-reflected link still accounts for the dominant portion of the received power. Therefore, the overall performance trend in Fig. \ref{EDoFv2} is still mainly determined by the proposed near-field IRS design rather than by the direct link itself.

\section{Conclusion}
In this paper, we investigated an IRS-assisted sparse MIMO system and showed that near-field propagation, when exploited by sparse arrays, can serve as a key enabler for effective multi-user spatial multiplexing. Our analysis revealed that near-field effects can substantially alleviate the user-channel correlation bottleneck inherent in far-field models, and we derived a closed-form favorable propagation metric and a geometry-driven deployment rule to guide system design. Furthermore, we showed that deterministic engineering of the IRS phases shapes the user channels into a well-conditioned form that supports robust interference management. For practical Rician fading channels, we developed an efficient alternating optimization algorithm to maximize the ergodic sum-rate by jointly designing the IRS phase shifts and user power allocation based on statistical CSI. Numerical results validated the analysis and demonstrated substantial performance gains, indicating that near-field sparse array engineering is a promising approach for enhancing the spatial multiplexing capability of IRS-assisted communications.  Future work will extend the proposed framework to account for mutual coupling, calibration errors, IRS hardware degradation, and wideband frequency-dependent near-field modeling.

\section*{APPENDIX A}
We derive the variance of the inner product $Y = \mathbf{h}_{j}^{H}\mathbf{h}_k$ for $j \neq k$. The variance is defined as $\mathrm{var}\{Y\} = \mathbb{E}[|Y|^2] - |\mathbb{E}[Y]|^2$.

The inner product can be expanded as $Y = \sum_{n_1, n_2} T_{n_1, n_2} e^{-\jmath\phi_{n_1}} e^{\jmath\phi_{n_2}}$, where $T_{n_a, n_b} \triangleq [\mathbf{g}_j]_{n_a}^* [\mathbf{F}^H\mathbf{F}]_{n_a,n_b} [\mathbf{g}_k]_{n_b}$. The second moment $\mathbb{E}[|Y|^2]$ is given by the four-fold summation:
\begin{align}
\mathbb{E} [|Y|^2]\!=\!\sum_{n_1,n_2,n_3,n_4}{T_{n_1,n_2}}T_{n_3,n_4}^{*}\mathbb{E} [e^{\jmath\psi \left( n_1,n_2,n_3,n_4 \right)}],
\end{align}
where the phase argument is $\psi \left( n_1,n_2,n_3,n_4 \right) =-\phi _{n_1}+\phi _{n_2}+\phi _{n_3}-\phi _{n_4}$. Given that the phase shifts $\{\phi_n\}$ are i.i.d. and drawn from $\mathcal{U}[0, 2\pi)$, the expectation of the complex exponential is non-zero (and equals 1) if and only if the indices are paired such that the net coefficient of each independent phase variable in the exponent is zero. This condition admits two surviving scenarios for the index pairings:
\begin{enumerate}
    \item[(a)] $n_1 = n_2$ and $n_3 = n_4$.
    \item[(b)] $n_1 = n_3$ and $n_2 = n_4$.
\end{enumerate}
Using the principle of inclusion-exclusion, the second moment is the sum of contributions from case (a) and case (b), minus the contribution from their intersection (where $n_1=n_2=n_3=n_4$).

The contribution from case (a) corresponds to $\sum_{n_1, n_3} T_{n_1,n_1}T_{n_3,n_3}^* = |\sum_{n=1}^N T_{n,n}|^2 = |\mathbb{E}[Y]|^2$.
The contribution from case (b) corresponds to $\sum_{n_1, n_2} T_{n_1,n_2}T_{n_1,n_2}^* = \sum_{n_1, n_2} |T_{n_1,n_2}|^2$.
The intersection is the case $n_1=n_2=n_3=n_4$, whose contribution is $\sum_{n=1}^N |T_{n,n}|^2$.

Combining these yields the second moment:
\begin{align}
\mathbb{E}[|Y|^2] &= |\mathbb{E}[Y]|^2 + \sum_{n_1, n_2} |T_{n_1,n_2}|^2 - \sum_{n=1}^N |T_{n,n}|^2 \nonumber \\
&= |\mathbb{E}[Y]|^2 + \sum_{n_1 \neq n_2} |T_{n_1,n_2}|^2.
\end{align}
The variance is then found by subtracting $|\mathbb{E}[Y]|^2$:
\begin{align}
\mathrm{var}\{\mathbf{h}_{j}^{H}\mathbf{h}_k\} = \sum_{n_1 \neq n_2} |T_{n_1, n_2}|^2.
\end{align}
Substituting the definition of $T_{n_1, n_2}$ gives the final result:
\begin{align}
\mathrm{var}\{\mathbf{h}_{j}^{H}\mathbf{h}_k\} = \sum_{n_1 \neq n_2} |[\mathbf{g}_j]_{n_1}|^2 |[\mathbf{g}_k]_{n_2}|^2 |[\mathbf{F}^H\mathbf{F}]_{n_1,n_2}|^2.
\end{align}
\section*{APPENDIX B}
We provide the derivations for the statistical moments used in the ergodic rate analysis. The cascaded channel is decomposed as $\mathbf{h}_k = \bar{\mathbf{h}}_k + \tilde{\mathbf{h}}_k$, comprising a deterministic LoS component $\bar{\mathbf{h}}_k = \sqrt{\alpha_k\beta_k} \mathbf{F} \mathbf{\Theta} \bar{\mathbf{g}}_k$ and a zero-mean NLoS component $\tilde{\mathbf{h}}_k = \sqrt{\alpha_k(1-\beta_k)} \mathbf{F} \mathbf{\Theta} \tilde{\mathbf{g}}_k$, where $\mathbb{E}[\tilde{\mathbf{g}}_k \tilde{\mathbf{g}}_k^H] = \mathbf{I}_N$. To facilitate the derivation with respect to the IRS phase shifts $\boldsymbol{\theta} = [e^{j\phi_1}, \dots, e^{j\phi_N}]^T$, we utilize the algebraic identity $\mathbf{\Theta}\mathbf{v} = \mathrm{diag}(\mathbf{v})\boldsymbol{\theta}$ for any vector $\mathbf{v}$.

The first moment $\mathbb{E}[\mathbf{h}_k^H\mathbf{h}_k]$ represents the average channel power. Leveraging the linearity of expectation and the fact that cross-terms vanish due to $\mathbb{E}[\tilde{\mathbf{h}}_k]=\mathbf{0}$, we have
\begin{align}
\mathbb{E}[\mathbf{h}_k^H\mathbf{h}_k] &= \|\bar{\mathbf{h}}_k\|^2 + \mathbb{E}[\|\tilde{\mathbf{h}}_k\|^2] \nonumber \\
&= \alpha_k\beta_k \|\mathbf{F}\mathrm{diag}(\bar{\mathbf{g}}_k)\boldsymbol{\theta}\|^2 + \alpha_k(1-\beta_k)\mathrm{tr}\left( \mathbf{F}\mathbf{\Theta}\mathbf{\Theta}^H\mathbf{F}^H \right) \nonumber \\
&= \alpha_k\beta_k \boldsymbol{\theta}^H \mathbf{A}_{k,k} \boldsymbol{\theta} + C_k, \label{eq:first_moment}
\end{align}
where $\mathbf{A}_{k,k} \triangleq \mathrm{diag}(\bar{\mathbf{g}}_k)^H \mathbf{F}^H \mathbf{F} \mathrm{diag}(\bar{\mathbf{g}}_k)$ and $C_k \triangleq \alpha_k(1-\beta_k)\|\mathbf{F}\|_F^2$.

Next, we evaluate the second moment $\mathbb{E}[|\mathbf{h}_k^H\mathbf{h}_j|^2]$ for $j \neq k$. Since $\tilde{\mathbf{h}}_k$ and $\tilde{\mathbf{h}}_j$ are independent zero-mean complex Gaussian vectors, terms involving an odd number of stochastic components evaluate to zero. Consequently, only four terms persist in the expectation:
\begin{align}
\mathbb{E}[|\mathbf{h}_k^H\mathbf{h}_j|^2] &= |\bar{\mathbf{h}}_k^H\bar{\mathbf{h}}_j|^2 + \mathbb{E}[|\bar{\mathbf{h}}_k^H\tilde{\mathbf{h}}_j|^2] \nonumber \\
&+ \mathbb{E}[|\tilde{\mathbf{h}}_k^H\bar{\mathbf{h}}_j|^2] + \mathbb{E}[|\tilde{\mathbf{h}}_k^H\tilde{\mathbf{h}}_j|^2].
\end{align}
The first term is purely deterministic and can be rewritten as:
\begin{align}
|\bar{\mathbf{h}}_k^H\bar{\mathbf{h}}_j|^2 &= \alpha_k\alpha_j\beta_k\beta_j |\boldsymbol{\theta}^H \mathrm{diag}(\bar{\mathbf{g}}_k)^H \mathbf{F}^H \mathbf{F} \mathrm{diag}(\bar{\mathbf{g}}_j) \boldsymbol{\theta}|^2 \nonumber \\
&= \alpha_k\alpha_j\beta_k\beta_j |\boldsymbol{\theta}^H \mathbf{A}_{k,j} \boldsymbol{\theta}|^2.
\end{align}
For the cross-terms, we apply the property $\mathbb{E}[|\mathbf{a}^H\mathbf{z}|^2] = \mathbf{a}^H \mathbb{E}[\mathbf{z}\mathbf{z}^H]\mathbf{a}$. Substituting the covariance $\mathbb{E}[\tilde{\mathbf{h}}_j\tilde{\mathbf{h}}_j^H] = \alpha_j(1-\beta_j)\mathbf{F}\mathbf{F}^H$, we obtain:
\begin{align}
% \mathbb{E}[|\bar{\mathbf{h}}_k^H\tilde{\mathbf{h}}_j|^2] &= \bar{\mathbf{h}}_k^H \mathbb{E}[\tilde{\mathbf{h}}_j\tilde{\mathbf{h}}_j^H] \bar{\mathbf{h}}_k \nonumber \\
% &= \alpha_k\alpha_j\beta_k(1-\beta_j) \boldsymbol{\theta}^H \mathbf{D}_k \boldsymbol{\theta},
\!\!\!\!\mathbb{E}[|\bar{\mathbf{h}}_k^H\tilde{\mathbf{h}}_j|^2] \!\!= \!\!\bar{\mathbf{h}}_k^H \mathbb{E}[\tilde{\mathbf{h}}_j\tilde{\mathbf{h}}_j^H] \bar{\mathbf{h}}_k\!\!=\!\! \alpha_k\alpha_j\beta_k(1-\beta_j) \boldsymbol{\theta}^H \mathbf{D}_k \boldsymbol{\theta},
\end{align}
where $\mathbf{D}_k \triangleq \mathrm{diag}(\bar{\mathbf{g}}_k)^H \mathbf{F}^H (\mathbf{F}\mathbf{F}^H) \mathbf{F} \mathrm{diag}(\bar{\mathbf{g}}_k)$. By symmetry, the third term is $\mathbb{E}[|\tilde{\mathbf{h}}_k^H\bar{\mathbf{h}}_j|^2] = \alpha_k\alpha_j\beta_j(1-\beta_k) \boldsymbol{\theta}^H \mathbf{D}_j \boldsymbol{\theta}$.

The final NLoS-NLoS interaction term involves the fourth-order moment. Using the standard property $\mathbb{E}[|\mathbf{x}^H\mathbf{y}|^2] = \mathrm{tr}(\mathbb{E}[\mathbf{x}\mathbf{x}^H]\mathbb{E}[\mathbf{y}\mathbf{y}^H])$ for independent Gaussian vectors, we derive:
\begin{align}
\mathbb{E}[|\tilde{\mathbf{h}}_k^H\tilde{\mathbf{h}}_j|^2] &= \mathrm{tr}\left( \alpha_k(1-\beta_k)\mathbf{F}\mathbf{F}^H \cdot \alpha_j(1-\beta_j)\mathbf{F}\mathbf{F}^H \right) \nonumber \\
&= \alpha_k\alpha_j(1-\beta_k)(1-\beta_j) \mathrm{tr}((\mathbf{F}\mathbf{F}^H)^2).
\end{align}
Summing these components yields the expressions in \eqref{E1} and \eqref{E2}.

\bibliographystyle{IEEEtran}
\bibliography{IEEEabrv,ref}
\end{document}